\documentclass[11pt,reqno]{amsart}
\usepackage{helvet}
\usepackage{enumitem}
\setlist[enumerate]{font=\normalfont,labelindent=*,leftmargin=*}
\setlist[itemize]{labelindent=*,leftmargin=*}
\setlist[description]{labelindent=*,leftmargin=*,itemindent=-1 em}
\usepackage{microtype}
\usepackage{color}
\usepackage{hyperref}

\count255=\the\catcode`\@ \catcode`\@=11 \edef\catc@de{\the\count255}

\newif\ifsuPer \suPertrue

\def\titlecomment#1{\def\@titlecomment{#1}}
\let\@titlecomment=\@empty

\renewcommand*\subjclass[2][2012]{%
  \def\@subjclass{#2}%
  \@ifundefined{subjclassname@#1}{%
    \ClassWarning{\@classname}{Unknown edition (#1) of ACM
      Subject Classification; using '2012'.}%
  }{%
    \@xp\let\@xp\subjclassname\csname subjclassname@2012\endcsname
  }%
}
\@namedef{subjclassname@1998}{1998 ACM Subject Classification}
\@namedef{subjclassname@2012}{2012 ACM CCS}

\def\ACMCCS#1{\def\@ACMCCS{#1}}
\let\@ACMCCS=\@empty

    \@xp\let\@xp\subjclassname\csname subjclassname@2012\endcsname

\def\amsclass#1{\def\@amsclass{#1}}
\let\@amsclass=\@empty
\def\amsclassname{2010 Mathematics Subject Classification}

\renewenvironment{abstract}{%
  \ifx\maketitle\relax
    \ClassWarning{\@classname}{Abstract should precede
      \protect\maketitle\space in AMS documentclasses; reported}%
  \fi
  \global\setbox\abstractbox=\vtop \bgroup
    \normalfont\Small
    \list{}{\labelwidth\z@
      \leftmargin3pc \rightmargin\leftmargin
      \listparindent\normalparindent \itemindent\z@
      \parsep\z@ \@plus\p@
      
    }%
    \item[\hskip\labelsep\scshape\abstractname.]%
}{%
  \endlist\egroup
  \ifx\@setabstract\relax \@setabstracta \fi
}
\newcommand{\revisionname}{Revision Note}
\newbox\revisionbox

\def\@setrevision{\@setrevisiona \global\let\@setrevision\relax}
\def\@setrevisiona{%
  \ifvoid\revisionbox
  \else
    \skip@20\p@ \advance\skip@-\lastskip
    \advance\skip@-\baselineskip \vskip\skip@
    \box\revisionbox
    \prevdepth\z@ 
    \bigskip\hrule\medskip
  \fi
}
\def\@setACMCCS{%
  {\itshape\subjclassname:}\enspace\@ACMCCS\@addpunct.}
\def\@setamsclass{%
  {\itshape\amsclassname:}\enspace\@amsclass\@addpunct.}
\def\@setkeywords{%
  {\itshape \keywordsname:}\enspace\@keywords\@addpunct.}
\def\@settitlecomment{\@titlecomment\@addpunct.}

\def\@setaddresses{\par
  \nobreak \begingroup
\footnotesize
  \def\author##1{\nobreak\addvspace\bigskipamount}%
  \def\\{\unskip, \ignorespaces}%
  \interlinepenalty\@M
  \def\address##1##2{\begingroup
    \par\addvspace\bigskipamount\noindent\narrower
    \@ifnotempty{##1}{(\ignorespaces##1\unskip) }%
    {\ignorespaces##2}\par\endgroup}%
  \def\curraddr##1##2{\begingroup
    \@ifnotempty{##2}{\nobreak\indent{\itshape Current address}%
      \@ifnotempty{##1}{, \ignorespaces##1\unskip}\/:\space
      ##2\par}\endgroup}%
  \def\email##1##2{\begingroup
    \@ifnotempty{##2}{\nobreak\indent{\itshape e-mail address}%
      \@ifnotempty{##1}{, \ignorespaces##1\unskip}\/:\space
      {##2}\par}\endgroup}%
  \def\urladdr##1##2{\begingroup
    \@ifnotempty{##2}{\nobreak\indent{\itshape URL}%
      \@ifnotempty{##1}{, \ignorespaces##1\unskip}\/:\space
      \ttfamily##2\par}\endgroup}%
  \addresses
  \endgroup
}
\copyrightinfo{}{}

\newinsert\copyins
\skip\copyins=3pc
\count\copyins=1000 
\dimen\copyins=.5\textheight 

\def\(#1{\relax\hbox to1.4 em{\rm(\hss#1\hss)}}
\RequirePackage{rotating}
\newdimen\cCHardim

\usepackage[linesnumbered,ruled,vlined]{algorithm2e}
\usepackage{amsmath}
\usepackage{amsfonts}

\DeclareMathOperator*{\argmin}{arg\,min}

\newcommand{\rr}{\mathbb{R}}
\newcommand{\zz}{\mathbb{Z}}
\newcommand{\nn}{\mathbb{N}}
\newcommand{\mech}{\mathcal{M}}
\newcommand{\dat}{\mathcal{X}}
\newcommand{\stack}{\textup{stack}}
\newcommand{\diag}{\textup{diag}}
\newcommand{\geoprop}{\beta}
\newcommand{\geopropStrat}{\gamma}
\newcommand{\queryprop}{\alpha}
\newcommand{\nc}{c}
\newcommand\numberthis{\addtocounter{equation}{1}\tag{\theequation}}

\newtheorem{theorem}{Theorem}
\newtheorem{observation}{Observation}
\newtheorem{lemma}{Lemma}

\newtheorem{example}{Example}
\newtheorem{definition}{Definition}

\newcommand{\upstairs}[1]{\textsuperscript{#1}}
\newcommand{\affilone}{\dag}
\newcommand{\affiltwo}{\ddag}
\newcommand{\affilthree}{$\diamond$}
\newcommand{\affilfour}{$\bigtriangleup$}
\newcommand{\affilfive}{$\nabla$}

\newcommand\emails[1]{%
  \begingroup
  \renewcommand\thefootnote{}\footnote{#1}%
  \addtocounter{footnote}{-1}%
  \endgroup
}

\date{March 1, 2024}
\title{Geographic Spines in the 2020 Census Disclosure Avoidance System}

\hypersetup{colorlinks=true,linkcolor=blue}

\usepackage{multicol}

\begin{document}

\begin{center}
  \maketitle

  \begin{tabular}{cc}
    Ryan Cumings-Menon,\upstairs{\affilone}
    Robert Ashmead,\upstairs{\affilone, \affiltwo}
    Daniel Kifer,\upstairs{\affilone,\affilthree}\\
    Philip Leclerc,\upstairs{\affilone}
    Jeffrey Ocker,\upstairs{\affilfour}
    Michael Ratcliffe,\upstairs{\affilone}
    Pavel Zhuravlev,\upstairs{\affilone}
    John M. Abowd\upstairs{\affilfive} \\
    
  \\[0.25ex]
   {\small \upstairs{\affilone} U.S. Census Bureau} \\
   {\small \upstairs{\affiltwo} The Ohio Colleges of Medicine Government Resource Center} \\
   {\small \upstairs{\affilthree} Penn State University} \\
   {\small \upstairs{\affilfour} Federal Communications Commission, formerly U.S. Census Bureau} \\
   {\small \upstairs{\affilfive} Cornell University, U.S. Census Bureau (retired)} \\
  \end{tabular}
  
   \emails{
    \upstairs{*}We thank Randall Akee and Norman DeWeaver for providing summary metrics that motivated the exploration of alternative geographic spines within the DAS \cite{NAP25978}. We also thank Simson Garfinkel and Nicholas Nagle for early contributions to the concept of using alternative geographic spines within the DAS. The authors are grateful for helpful comments from Justin Doty, Brian Finley, Mark Fleischer, Sallie Keller, and the participants of the DIMACS and Rutgers University Department of Statistics' Workshop on the Analysis of Census Noisy Measurement Files and Differential Privacy on April $27-29,$ 2022. The views expressed in this technical paper are those of the authors and not those of the U.S. Census Bureau. The Census Bureau has reviewed this data product to ensure appropriate access, use, and disclosure avoidance protection of the confidential source data (Project No. 7502798, Disclosure Review Board (DRB) approval number: CBDRB-FY24-CED005-0001).
    }
\end{center}

\begin{abstract}
    The 2020 Census Disclosure Avoidance System (DAS) is a formally private mechanism that first adds independent noise to cross tabulations for a set of pre-specified hierarchical geographic units, which is known as the geographic spine. After post-processing these noisy measurements, the DAS outputs a formally private database with fields indicating location in the standard census geographic spine, which is defined by the United States as a whole, states, counties, census tracts, block groups, and census blocks. This paper describes how the geographic spine used internally within the DAS to define the initial noisy measurements impacts accuracy of the output of the DAS. Specifically, tabulations for geographic areas tend to be most accurate for geographic areas that both 1) can be derived by aggregating together the geographic units of the internal spine other than block geographic units, and 2) are closer to the geographic units of the internal spine. After describing the methods supported by the DAS for defining the internal DAS geographic spine, we provide the settings used to define the 2020 Census production DAS executions. To demonstrate the accuracy impact of the choice of internal spine, we also provide accuracy metrics for three DAS executions using settings similar to the ones used for the 2020 Census Redistricting Data (P.L. 94-171) persons production DAS execution, but using three different choices of internal geographic spine.
\end{abstract}

\vspace*{0.15in}
\hspace{10pt}

\section{Introduction}

The U.S. Census Bureau's 2020 Disclosure Avoidance System (DAS) was used to produce a formally private file in support of the 2020 Census Redistricting Data (P.L. 94-171) Summary File (hereafter redistricting data) and the 2020 Demographic and Housing Characteristics File (DHC). As described by \cite{Abowd:et:al:2022,cumings2023disclosure} in more detail, the DAS takes the confidential Census Edited File (CEF) as input, which is a set of microdata files with components that include data on individuals, housing units, group quarters (GQs), and households; the output of the DAS is called the Microdata Detailed File (MDF). We call the DAS formally private because it provides privacy guarantees based on either differential privacy (DP) \cite{dwork2006calibrating} or zero-concentrated differential privacy (zCDP) \cite{bun2016concentrated}, as we describe in Section \ref{prelims} in more detail.

Each DAS execution begins by converting the microdata in the CEF to a histogram, \textit{i.e.}, a flattened and fully saturated contingency table, for each census block. Afterward, histograms for less granular \textit{geographic units}, or \textit{geounits}, are created by adding these block level histograms to one another. For example, in early versions of the DAS, these block histograms were added to one another to define the histograms for each census block group. This recursion was repeated to define the histograms for each census tract, county, and state geographic units, and also, in US DAS executions, a final geounit corresponding to the U.S. as a whole.\footnote{ Aside from US DAS executions discussed in this paragraph, there are also Puerto Rico (PR) DAS executions, which have a least granular geounit at the state geographic level that corresponds to the  Commonwealth of Puerto Rico as a whole.} The collection of all of these hierarchical geographic units is known as a \textit{geographic spine}, or a \textit{spine}, and this particular spine is known as the conventional geographic spine, as defined by the Census Bureau Geography Division. The geounits in the geographic spine are split into several levels of granularity, which are known as \textit{geographic levels}, or \textit{geolevels}; for example, each geounit in the conventional geographic spine is a member of either the U.S., state, county, census tract, block group, or the census block geolevels. In other words, the spine can be viewed as a rooted tree, with a root vertex given by the U.S. geounit. This connection with rooted trees also motivates the terms we use to describe adjacency relationships on the spine, such as parent and child geounits. For example, for the conventional spine, the child geounits of the U.S. geounit are the geounits in the state geographic level.

To see the importance of the internal spine used within the DAS, a general overview of what takes place during a DAS execution will be helpful. After the CEF-based histograms of each geounit on the internal spine are computed, the DAS uses a formally private implementing  mechanism to observe noisy measurements for each geounit, which are defined as tabulations, \textit{i.e.}, counts within various population groups within the geounit (such as population counts within each race group), added to realizations of mean-zero random variables. Afterward, the DAS uses the noisy measurements of  the  U.S. geounit, to compute a formally private data histogram for this geounit. Next, the DAS uses the noisy measurements of progressively more granular (lower) geolevels on the spine to estimate formally private data histograms for these geounits, subject to the constraint that these histograms are consistent with the estimates of the next-higher (parent) geolevels. After this recursion is run down to the block geolevel, the output MDF is defined by converting the block histogram estimates to microdata. For more detail on the DAS, including the optimization problems used to define the histogram estimates of each geounit, see \cite{Abowd:et:al:2022,cumings2023disclosure}.

Note that the main topic of this paper is the geographic spine that is used internally in the DAS. Even when this geographic spine does not correspond to the conventional spine, the published statistical tabulations based on the DAS output are still provided for geounits on the conventional spine, defined using the Census Bureau Geography Division's standard definitions.

\subsection{Contributions}

The conventional spine has two primary shortcomings when used for the internal spine in the DAS. First, legally defined American Indian/Alaska Native/Native Hawaiian (AIAN) tribal areas are usually far from this spine, meaning that many on-spine geounits must be added to or subtracted from one another to compose these geographic areas. Being distant from the spine often results in query count estimates with greater mean squared error than on-spine geounits. Second, the conventional spine also places many other important legal, political or census-designated off-spine entities (OSEs) far from the spine, \textit{e.g.}, minor civil divisions (MCDs) and incorporated places.

In order to ameliorate these two issues, the Census Bureau experimented with using alternative spines within the DAS. The purpose of this paper is to describe these alternative spines, the methods we used to define them, and the spines used within the DAS for the production 2020 redistricting and DHC DAS implementations. 

The starting point for these alternative spines is a spine called an \textit{AIAN spine}. The goal for these spines is to ameliorate the issue of AIAN tribal areas being far from the spine. To define an AIAN spine, first let the \textit{Aggregated AIAN Areas OSE} be defined as the collection of all blocks that are within the AIAN areas. An AIAN spine is defined by  replacing geounits in the state geolevel and below with up to two geounits. Specifically, each such original geounit is replaced by geounits defined as the portion of this geounit that is inside, and the portion that is outside, of the Aggregated AIAN Areas OSE.\footnote{In cases in which one of these two portions of the original geounit do not contain any blocks, the original geounit is replaced with a single geounit. For example, since there are not any AIAN areas in Kentucky, an AIAN spine would not include a geounit corresponding to the AIAN portion of Kentucky, but there would be a geounit that corresponds to the non-AIAN portion of Kentucky.} For example, the geounit corresponding to the state of Arizona would be replaced with two state geolevel geounits in an AIAN spine: 1) the portion of Arizona that is inside of the AIAN areas, and 2) the portion of Arizona that is outside of the AIAN areas. 

This process of splitting a spine into an AIAN and a non-AIAN branch at the state geolevel and below can be applied to a variety of different spines. In the case of the 2020 redistricting DAS implementation, we apply this process to the conventional spine; a graphical summary of the 2020 redistricting AIAN spine is shown in Figure \ref{fig:aian_spine}. The 2020 DHC AIAN spine is defined by the output of this same process applied to an alternative input spine, as described in Section \ref{settings} in more detail.

\begin{figure}
    \centering
    \includegraphics[width=90mm]{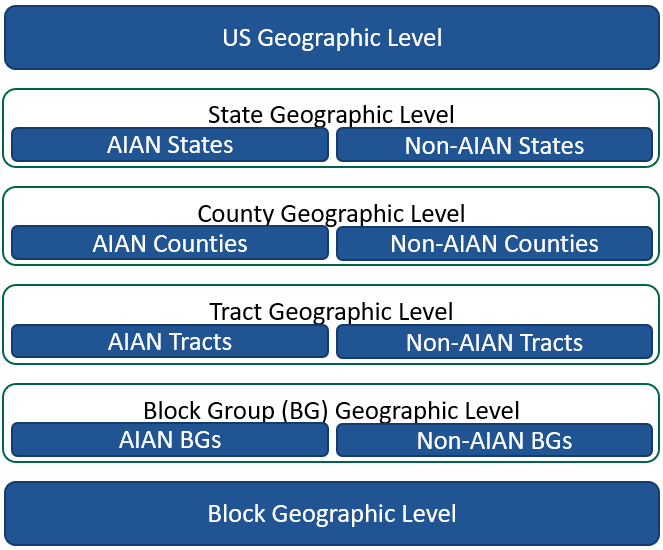}
    \caption{The redistricting AIAN spine, shown without the optional tract group geographic level between the tract and county geographic levels. This spine uses the standard census spine as a starting point and includes an AIAN branch. The optimized spine uses an AIAN spine as a starting point and redefines block groups as optimized block groups.}
    \label{fig:aian_spine}
\end{figure}

This paper describes several optimization heuristics that perform updates to an AIAN spine to enhance specific features of the spine. We call geographic spines that are output from these spine optimization methods \textit{optimized spines}. The spine optimization routines update an input spine in two stages. In the first stage, certain geolevels are redefined to bring a given set of OSEs closer to the spine. For example, all of the spine optimization methods the Census Bureau has experimented with involve redefining the block groups in an AIAN spine as \textit{optimized block groups}. We also describe methods that add additional geolevels, \textit{e.g.}, a geolevel between tracts and counties known as the tract group geolevel, to provide a further improvement in the accuracy of tabulations for the set of target OSEs. After the initial definitions of the geounits have been found in the first stage of the spine optimization routines, the variance of the final histogram estimate is reduced in the second stage of spine optimization by removing certain redundant geounits, which we call the bypassing step. For example, in cases in which a parent geounit has only one child geounit, this decision rule results in bypassing the parent geounit, \textit{i.e.}, removing the parent geounit from the spine, which, as we will see below in more detail, also allows us to decrease the variance of the noisy measurements of the child geounit. Since the consistency-with-parent constraints used in the DAS require that the count estimates of the child are equal to those of the parent in cases in which a parent geounit has only one child, bypassing the parent geounit in this way increases accuracy of the child's histogram estimate by reducing the variance of the noisy measurements that are used to derive this estimate.

In addition to describing the spine optimization routines that are supported by the DAS, we also provide the spine optimization settings of the 2020 production redistricting and DHC DAS executions. To demonstrate the differences in the accuracy provided by the conventional spine, an AIAN spine, and an optimized spine, we compare accuracy metrics for DAS executions that use each of these three spines. 

\subsection{Paper Outline}

The heuristics used in the spine optimization routines are based on existing results from the DP literature, which are provided in Section \ref{prelims}. Section \ref{def_spine} describes how geographic spines are represented in this paper using matrices and provides a result on the privacy guarantees of the DAS when the internal spine is defined as either the conventional spine or an AIAN spine. The methods used in the first stage of spine optimization, \textit{i.e.}, the methods that redefine certain geolevels to bring a given set of OSEs closer to the spine, are described in Section \ref{osed}. Section \ref{bypassing} describes the bypassing method used in the second stage and provides a result on the privacy guarantees of the DAS when an optimized spine is used. Section \ref{settings} describes the geographic spines that were used in the 2020 production DAS executions. Section \ref{sec:metrics} concludes the paper with a comparison of accuracy metrics for redistricting persons DAS executions with an internal spine defined by either the conventional spine, an AIAN spine, or an optimized spine. The next subsection outlines the notation used throughout this paper.

\subsection{Notation}

Throughout the paper we denote matrices using capitalized font, and vectors using lower case bold font. We denote the Kronecker product of real matrices $A,B$ by $A\otimes B,$ and elementwise division, when these matrices are conformable, by $A \oslash B$ respectively. We denote the elementwise absolute value of the real matrix $A$ by $\lvert A \rvert.$ A length $N$ column vector with each element equal to $c\in\rr$ is denoted by $\boldsymbol{c}_N,$ and, when there is little risk of confusion, we omit the subscript $N.$ We use $I_n$ to denote the $n\times n$ identity matrix. We denote the $i^\textrm{th}$ row of $A\in \rr^{m\times n}$ by $A[i,\cdot]$ and the $i^\textrm{th}$ column by $A[\cdot,i].$ Given two datasets (possibly with unequal numbers of rows) $x, x^\prime,$ let $d_{\mathcal{H}}(x, x^\prime)$ denote the minimum number of set addition and subtraction operations with singleton sets required to transform the dataset $x$ to be equal to the dataset $x'.$ We also denote the set of nonnegative real numbers by $\rr_+,$ the positive reals by $\rr_{++},$ and the nonnegative integers by $\zz_+.$ The number of elements in the set $A$ will be denoted by $\textup{Card}(A).$

We make use of several probability distributions, but two are worth pointing out explicitly, as their definitions are less standard. First, let the discrete Gaussian random variable be defined so that its probability mass function, $f:\zz \rightarrow \rr,$ is given by $f(x) \propto \exp(-(x-\mu)^2/(2 \sigma^2)),$ and denote the distribution of this random variable by $\textup{N}_\zz(\mu, \sigma^2);$ see for example, \cite{canonne2020discrete}. In contrast, to denote the (continuous) Gaussian distribution we use the standard notation, $\textup{N}(\mu, \sigma^2).$ We also define the discrete Laplace random variable so that its probability mass function, $f:\zz\rightarrow \rr,$ is given by $f(x) \propto \exp(-\lvert x-\mu \rvert/b),$ and denote the distribution of this random variable by $\textup{Laplace}_\zz(\mu, \sigma^2);$ see for example, \cite{ghosh2012universally}. We also denote the distributions of length $n$ column vectors with each element $i\in\{1,\dots,n\}$ independently distributed as $\textup{N}_\zz(\boldsymbol{\mu}[i], \boldsymbol{c}[i])$, $\textup{Laplace}(\boldsymbol{\mu}[i], \boldsymbol{c}[i])$, and $\textup{Laplace}_\zz(\boldsymbol{\mu}[i], \boldsymbol{c}[i]),$ where $\boldsymbol{\mu}\in\rr^n$ and $\boldsymbol{c}\in \rr^n_+,$ by $\textup{N}_\zz(\boldsymbol{\mu},\diag(\boldsymbol{c})),$ $\textup{Laplace}(\boldsymbol{\mu},\diag(\boldsymbol{c})),$ and $\textup{Laplace}_\zz(\boldsymbol{\mu},\diag(\boldsymbol{c}))$ respectively.

\section{Differential Privacy} \label{prelims}
\subsection{Privacy Definitions}

The definitions of differential privacy and $\rho$--zCDP are provided by \cite{dwork2006our,dwork2006calibrating,bun2016concentrated}, and these definitions are also restated below. Both definitions can be viewed as a bound on the difference between the distribution of a mechanism's output when the input is a given database and when the input is a \textit{neighboring database}. The specific notion of difference used in these definitions, along with their required upper bounds, is generally self-evident, but more care is required for the definition of a neighboring database because there are multiple definitions for this term that are commonly used. For example, two databases are said to be \textit{unbounded neighbors} if one database can be derived by adding \textit{or} removing a record from the other database; likewise, two databases are said to be \textit{bounded neighbors} if one database can be derived from the other by adding one record \textit{and} removing one record \cite{kifer2011no}. In other words, the set of bounded neighbors can be written as $ \{x, x^\prime \in \dat^d \mid  d_{\mathcal{H}}(x, x^\prime) = 2 \},$ where $\dat^d$ is the set of databases with $d$ records, and the set of unbounded neighbors as $\{x, x^\prime \mid  d_{\mathcal{H}}(x, x^\prime) = 1 \}.$

The definition of neighbor has important implications because the upper bounds provided in the privacy guarantees below can be translated into an upper bound on the power of any statistical test with a given significance level of the null hypothesis that the input database is $x,$ versus the alternative that it is a neighbor of $x$ \cite{dong2022gaussian, wasserman2010statistical}. In other words, using bounded neighbors does not provide a theoretical bound on the accuracy of inferences on the number of records in the input database. However, in the case of tests with a null hypothesis that the database is $x$ and an alternative hypothesis that it is a given bounded neighbor of $x,$ formal privacy definitions that use bounded neighbors provide stronger limitations on the accuracy of privacy attackers' inferences. This can be shown for the privacy definitions below using the fact that a bounded neighbor of a given database is also an ``unbounded neighbor of an unbounded neighbor" of the database.

The privacy guarantees provided by the DAS use a unique definition of neighbor that does not protect against certain inferences. Specifically, the DAS does not protect each state's total population or the locations of each group quarters type. In other words, the DAS outputs a database with state total populations that are identical to those of the input database, and, since the frame of a U.S. decennial census excludes addresses that are unoccupied group quarters, at least one person must reside in each group quarter that appears in the census.\footnote{That is, the definition of a living quarter eligible for inclusion in the census of population and housing does not include vacant group quarters. By contrast, housing units may be either occupied or vacant.} The population residing in each group quarter type in a given geounit is constrained to be at least the number of group quarters of that type within the geounit. In addition to these data-dependent constraints, the DAS also imposes data-independent constraints called imputation and edit constraints. For example, respondents that reside in college/university student housing are required to be between the ages of 16 and 65 and respondents in nursing facilities are required to be at least 20 years of age. The imputation and edit constraints imposed on the CEF ensure that each row of the CEF satisfies these constraints, so we also impose these data-independent imputation and edit constraints on the output MDF produced by the DAS. For more detail on the data-dependent and data-independent constraints imposed by the DAS, see \cite{Abowd:et:al:2022,cumings2023disclosure}. 

For this reason, in the context of the DAS, two databases are said to be neighbors if one database can be derived from the other by adding one record to a given state and then removing one record from that same state, such that both databases also satisfy the group quarters invariant constraints and the imputation and edit constraints of the input database of the DAS. To emphasize this definition of neighbors in the presence of invariant constraints in our privacy definitions below, we define the set of pairs of neighboring databases as $\{x, x^\prime \in  \mathcal{U} \mid  d_{\mathcal{H}}(x, x^\prime) = 2 \},$ where $\mathcal{U}$ is a given universal set of databases. For example, this notation is a generalization of the definition of bounded neighbors, which used $\mathcal{U}=\dat^{d}.$ In the results throughout this paper that use these privacy definitions, we denote the set of databases that satisfy these group quarter invariant constraints as well as the imputation and edit constraints of the confidential input data as $\mathcal{G}$ and the set of databases with $\boldsymbol{d}[i]$ records in each state $i$ as $\dat^{\boldsymbol{d}}.$ Using this notation, the definition of neighbors used by the DAS can be written as, $\{x, x^\prime \in \dat^{\boldsymbol{d}} \cap \mathcal{G} \mid  d_{\mathcal{H}}(x, x^\prime) = 2 \}.$

\begin{definition}\label{dp_def}
(Differential Privacy \cite{dwork2006calibrating}) A randomized algorithm $\mech: \mathcal{U} \rightarrow \mathcal{Y}$ satisfies $\epsilon$--differential privacy if, for all $(x, x^\prime) \in \{x, x^\prime \in  \mathcal{U} \mid  d_{\mathcal{H}}(x, x^\prime) = 2 \}$ and all $E \subset \mathcal{Y},$ we have $\textup{P}(\mech(x) \in E) \leq \exp(\epsilon) \textup{P}(\mech(x^\prime) \in E).$
\end{definition}

The definition of $\rho$--zCDP \cite{bun2016concentrated} is provided below; this privacy framework is used for privacy-loss budget (PLB) accounting in both the 2020 redistricting data file and the DHC production DAS executions. Kifer et al. \cite{kifer2022bayesian} provide more information on privacy semantic statements that can be made for the DAS.

\begin{definition}\label{zcdp_def}
    (Zero-Concentrated Differential Privacy (zCDP) \cite{bun2016concentrated}) A randomized mechanism $\mech : \mathcal{U} \rightarrow \mathcal{Y}$ is $\rho$--zero-concentrated differentially private ($\rho$--zCDP) if, for all $x, x' \in  \{x, x^\prime \in  \mathcal{U} \mid  d_{\mathcal{H}}(x, x^\prime) = 2 \}$ and all $\alpha \in (1,\infty)$,
    $$D_\alpha(\mech(x)||\mech(x')) \le \rho \alpha, $$
    where $D_\alpha(P||Q) = \log \left(\sum_{E\in \mathcal{Y}} P(E)^\alpha Q(E)^{(1-\alpha)} \right)/(\alpha-1)$ is the R\'enyi divergence of order $\alpha$ between the distributions $P$ and $Q$.
\end{definition}

\subsection{Privacy-Loss Accounting Results}

A general pattern in our proofs that the DAS is formally private for various spine settings is to first establish that the DP implementing  mechanism for each query $i$ is either $\epsilon_i$--DP or $\rho_i$--zCDP, and then use parallel and sequential composition to show that the combination of all of the DP implementing  query answers used in the DAS  are $\sum_i\epsilon_i$--DP or $\sum_i\rho_i$--zCDP. This section contains these intermediate results that are used in these proofs. The first two results are used in the next subsection to show that each implementing  mechanism is either $\epsilon$--DP or $\rho$--zCDP.

\begin{lemma}\label{laplace_dp_1d}
(Theorem 3.2 \cite{ghosh2012universally}) Let $\Delta, \epsilon > 0$.  Let $q: \mathcal{U} \rightarrow \zz$ satisfy $|q(x)-q(x')| \le \Delta$ for all $x,x' \in  \{x, x^\prime \in  \mathcal{U} \mid  d_{\mathcal{H}}(x, x^\prime) = 2 \}.$ Define a randomized algorithm by $\mech(x) = q(x) + Y$ where either $Y \sim \textup{Laplace}_{\zz}(0, \Delta/\epsilon)$ or $Y \sim \textup{Laplace}(0, \Delta/\epsilon).$  Then $\mech$ satisfies $\epsilon$--DP.
\end{lemma}

\begin{lemma}\label{g_cdp_1d}
(Theorem 4 \cite{canonne2020discrete}) Let $\Delta, \epsilon > 0$.  Let $q: \mathcal{U} \rightarrow \zz$ satisfy $|q(x)-q(x')| \le \Delta$ for all $x,x' \in  \{x, x^\prime \in  \mathcal{U} \mid  d_{\mathcal{H}}(x, x^\prime) = 2 \}.$ Define a randomized algorithm by $\mech(x) = q(x) + Y$ where either $Y \sim \textup{N}_{\zz}(0, \Delta^2/(2\rho))$ or $Y \sim \textup{N}(0, \Delta^2/(2\rho)).$ Then $\mech$ satisfies $\rho$--zCDP.
\end{lemma}

While providing a comprehensive compilation of the properties that are implied by a mechanism satisfying either $\epsilon$--DP or $\rho$--zCDP is beyond the scope of this paper, a few of these properties are worth explicitly pointing out. First, these privacy guarantees share the property that they are invariant to post-processing. In other words, if $\mech : \mathcal{U} \rightarrow \mathcal{Y}$ satisfies either of these privacy guarantees, then so does the mechanism $f \circ \mech,$ where $f:\mathcal{Y} \rightarrow \mathcal{Z}.$ Second, if $\mech_1(x)$ is $\epsilon_1$--DP (respectively, $\rho_1$--zCDP) and $\mech_2(x)$ is $\epsilon_2$--DP ($\rho_2$--zCDP) then releasing the output of $\mech_1(x)$ and $\mech_2(x)$ simultaneously is itself $(\epsilon_1+\epsilon_2)$--DP ($(\rho_1+\rho_2)$--zCDP), which is called sequential composition. Also, when $\mech_1(x),\mech_2(x),\dots,\mech_k(x)$ only depend on disjoint subsets of the input dataset $x,$ releasing the outputs of these mechanisms simultaneously is $(\max\;\{\epsilon_1,\epsilon_2\})$--DP (respectively, $(\max\;\{\rho_1,\rho_2\})$--zCDP) when using unbounded neighbors and $\max_{i\neq j} \epsilon_i+\epsilon_j$ (respectively, $(\max_{i\neq j} \rho_i+\rho_j)$--zCDP) when using bounded neighbors, which is known as parallel composition. The following lemma uses sequential and parallel composition to provide privacy guarantees for a mechanism that is defined by multiple univariate mechanisms described in Lemmas \ref{laplace_dp_1d} and \ref{g_cdp_1d}.

\begin{lemma}\label{composition}
(Theorem 14 \cite{canonne2020discrete}, Proposition 1 \cite{dwork2006our}, Theorem 3.2 \cite{ghosh2012universally})
Suppose $x,x' \in \mathcal{X}^d$ differ on a single entry, and let a randomized algorithm $\mech: \mathcal{X}^d \rightarrow \mathcal{Y}^m$ be defined by $\mech(x) = q(x) + \boldsymbol{y}$ where $q:\mathcal{X}^d \rightarrow \mathcal{Y}^m$ and $\boldsymbol{y}$ is an $m$ dimensional column vector of independent random variables. Then we have the following.
\begin{enumerate}
    \item Let $\boldsymbol{b}\in \rr^m_{++},$ $\epsilon>0,$ and, for all neighbors $x,x' \in \mathcal{X}^d,$ suppose $\sum_{j=1}^m \lvert q(x)[j] - q(x')[j]\rvert / \boldsymbol{b}[j] \le \epsilon.$
    If, for all $j \in \{1,\dots, m\},$ either $\boldsymbol{y}_{j} \sim \textup{Laplace}_{\zz}(\mathbf{0}, \boldsymbol{b}[j])$ or $\boldsymbol{y}_{j} \sim \textup{Laplace}(\mathbf{0}, \boldsymbol{b}[j]),$ then $\mech$ satisfies $\epsilon$--DP.
    \item Let $\boldsymbol{\sigma^2}\in\rr^m_{++},$ $\rho>0,$ and, for all neighbors $x,x' \in \mathcal{X}^d,$ suppose $\sum_{j=1}^m \frac{(q(x)[j] - q(x')[j])^2}{\boldsymbol{\sigma^2}[j]} \le 2 \rho.$
    If, for all $j \in \{1,\dots, m\},$ either $\boldsymbol{y}_{j} \sim \textup{N}_{\zz}(\mathbf{0}, \boldsymbol{\sigma^2}[j])$ or $\boldsymbol{y}_{j} \sim \textup{N}(\mathbf{0}, \boldsymbol{\sigma^2}[j]),$ then $\mech$ satisfies $\rho$--zCDP.
\end{enumerate}
\end{lemma}

\subsection{Linear Queries and Marginal Query Groups}

The notation for each query $q(\cdot)$ in the preceding subsections is more general than required to describe the DAS. Specifically, the DAS uses only linear queries. We define a \textit{linear query} as a linear map represented by a length $n$ vector, and a linear query answer as the inner product of a linear query and the histogram cell counts vector. Since we only consider linear queries from this point forward, we sometimes refer to a linear query simply as a \textit{query} below. We also define a \textit{query matrix} as the matrix representation of a collection of linear queries, \textit{i.e.}, a matrix defined by vertically stacking the length $n$ row vector representations of the collection of linear queries together.

All elements of the linear queries used within the DAS have elements that are either zero or one, and they are defined so that each linear query answer provides an individual count for a marginal of the full histogram. It is convenient to group the queries providing counts for the same marginal together, so we will define a \textit{query group} as the linear queries that provide counts for the same marginal, vertically stacked on top of one another. For example, suppose the schema of the original database is CENRACE $\times$ HISPANIC $\times$ VOTINGAGE, where CENRACE, HISPANIC, and VOTINGAGE indicate one of 63 census race combination categories, one of two ethnicity categories, and one of two age categories, respectively.\footnote{While we are using attributes that are part of the census, note that this example, and the queries in this section, are simply described in the context of a hypothetical survey that does not use a geographic spine. The linear query notation introduced here is extended in the next section to account for the presence of the geographic spine.} In this case, the CENRACE $\times$ VOTINGAGE query group matrix is defined as the matrix

$$
Q = I_{63} \otimes \boldsymbol{1}^\top_2 \otimes I_2,
$$

\noindent
and right multiplying $Q$ by the vector of histogram cell counts of a database provides the CENRACE $\times$ VOTINGAGE query group answers. More generally, the query matrix of a given query group is defined as a matrix that can be written as a series of Kronecker products of identity matrices and row vectors of ones. Note that this definition encompasses cases in which each of these Kronecker factors are all either row vectors of ones, \textit{i.e.}, the total population query group matrix, or are all identity matrices, \textit{i.e.}, the detailed cell query group matrix.

One of the inputs of the DAS is the proportion of PLB allocated to each query $i$ of a given geolevel, which we denote by $\queryprop_i\in \rr_{++}.$ The following lemma uses this notational convention, while only considering a single geounit, to provide the $\epsilon$--DP and $\rho$--zCDP privacy guarantees of releasing the output of the implementing  mechanisms of this geounit.\footnote{Alternatively, for datasets that contain an attribute identifying each respondent's geounit for a given geolevel, this theorem can also be used to provide the privacy guarantees of releasing the output of all of the implementing  DP mechanisms for the geolevel. However, depending on how the operation of bypassing geounits is defined, this can be less straightforward in the case of the optimized spine because there may be respondents in the dataset that are not included in any geounit in the geolevel. This alternative interpretation of the Lemma is discussed in more detail in Sections \ref{def_spine}, and the operation of bypassing geounits is defined in \ref{bypassing}.} This is used in the next section to provide the privacy guarantees of the DAS when either the conventional spine or an AIAN spine is used. Note that this lemma also makes use of a notational convention that we use throughout the rest of the paper; the histogram of detailed cell counts computed on the confidential dataset $x\in \dat^d$ is denoted by $\boldsymbol{x}.$

\begin{lemma}\label{per_geostrat_is_dp}
If each $Q[i]\in\{Q[i]\}_{i=1}^q$ is a query group matrix of dimension $m[i] \times n$ and the mechanism $\mech(\boldsymbol{x}),$ outputs $\{Q[i] \boldsymbol{x} + \boldsymbol{y}_{i}\}_{i=1}^q,$ where $\boldsymbol{x}$ is a vector of the histogram counts for $x\in\dat^d,$ and $\boldsymbol{y}_{i}$ is a length $m[i]$ column vector of independent random variables, then we have the following.
\begin{enumerate}
    \item If, for each $i\in\{1,\dots, q\},$ either 
    \begin{center}
        $\boldsymbol{y}_{i} \sim \textup{Laplace}(\mathbf{0}, \diag(\mathbf{2}_{m[i]}/(\epsilon\queryprop_i)))$ or $\boldsymbol{y}_{i} \sim \textup{Laplace}_\zz(\mathbf{0}, \diag(\mathbf{2}_{m[i]} /(\epsilon\queryprop_i) ) ),$
    \end{center}
    \noindent
    where $\sum_i\queryprop_i=1,$ then $\mech(\cdot)$ satisfies $\epsilon$--DP.
    \item If, for each $i\in\{1,\dots, q\},$ either 
    \begin{center}
    $\boldsymbol{y}_{i} \sim \textup{N}(\mathbf{0}, \diag(\mathbf{1}_{m[i]}/(\rho\queryprop_i)))$ or $\boldsymbol{y}_{i} \sim \textup{N}_\zz(\mathbf{0}, \diag(\mathbf{1}_{m[i]}/(\rho\queryprop_i))),$ 
    \end{center}
    \noindent
    where $\sum_i\queryprop_i=1,$ then $\mech(\cdot)$ satisfies $\rho$--zCDP.
\end{enumerate}
\end{lemma}
\begin{proof}
Since, for each $i\in \{1,\dots, q\},$ only one element of each column of $Q[i]$ is nonzero, we use Lemma \ref{composition} to leverage parallel composition to prove both cases. This reduces the problem to carrying out sequential composition of the set of mechanisms, $\{\mech[i](\boldsymbol{x})\}_i^q,$ where $\mech[i](\boldsymbol{x}) := Q[i] \boldsymbol{x} + \boldsymbol{y}_{i}.$ In both cases, we use the fact that $\lvert Q[i] \boldsymbol{x} - Q[i] \boldsymbol{x}^\prime \rvert,$ where $\boldsymbol{x},\boldsymbol{x}^\prime$ are vectors of histogram cell counts of databases differing on a single entry, is a vector with at most two elements that are equal to one, with the remaining elements equal to zero.

In the first case, since $\sum_j  \lvert Q[i][j,\cdot] \boldsymbol{x} - Q[i][j,\cdot] \boldsymbol{x}^\prime \rvert \leq 2,$ the first case in Lemma \ref{composition} implies $\mech[i](\boldsymbol{x})$ is $\queryprop_i\epsilon$--DP. Thus, using sequential composition, releasing the output of $\{\mech[i](\boldsymbol{x})\}_i$ is $\sum_i\queryprop_i \epsilon =\epsilon$--DP.

In the second case, since $\sum_j  \lvert Q[i][j,\cdot] \boldsymbol{x} - Q[i][j,\cdot] \boldsymbol{x}^\prime \rvert ^2 \leq 2,$ the second case in Lemma \ref{composition} implies $\mech[i](\boldsymbol{x})$ is $\queryprop_i\rho$--zCDP. Thus, using sequential composition, releasing the output of $\{\mech[i](\boldsymbol{x})\}_i$ is $\sum_i\queryprop_i \rho =\rho$--zCDP.
\end{proof}

\subsection{The Matrix Mechanism} \label{matmech}

This section describes a class of random mechanisms known as matrix mechanisms in the context of linear queries composed of query groups \cite{li2010optimizing}. Using the notation introduced in Lemma \ref{per_geostrat_is_dp}, we assume that the errors of the $i^{\textup{th}}$ query group are random variables distributed as either $\boldsymbol{y}_{i} \sim \textup{Laplace}(\mathbf{0}, \boldsymbol{2}_{m[i]} /(\epsilon\queryprop_i)),$ in the case of $\epsilon$--DP, or $\boldsymbol{y}_{i} \sim \textup{N}(\mathbf{0}, \boldsymbol{1}_{m[i]}/(\rho\queryprop_i)),$ in the case of $\rho$--zCDP, where $\queryprop_i\in\rr_{++}$ satisfies $\sum_i\queryprop_i = 1.$ Also, let the \textit{workload} matrix $W \in \{0,1\}^{m\times n}$ be defined as $W := \stack(\{Q[i]\}_{i=1}^q),$ the errors vector as $\boldsymbol{y} := \stack(\{\boldsymbol{y}_{i}\}_i^q),$ the PLB proportions vector as $\boldsymbol{\queryprop} := \stack(\{\mathbf{1}_{m[i]}\queryprop_i\}_{i=1}^q),$ and the response variable as $\boldsymbol{z} := W \boldsymbol{x} + \boldsymbol{y},$ which is also the vertically stacked output of the mechanism described in Lemma \ref{per_geostrat_is_dp}. We assume that $W$ has full column rank, which can be ensured by including the detailed cell query group in $W,$ \textit{i.e.}, the identity matrix.

For both of the distributions that we consider in this section, we define an alternative response variable that is observationally equivalent to $\boldsymbol{z},$ but with homoscedastic errors. For example, in the case of $\epsilon$--DP, we have $\boldsymbol{y} \sim \textup{Laplace}(\mathbf{0}, \diag(\boldsymbol{2}\oslash (\epsilon\boldsymbol{\queryprop}))),$ so we can define the rescaled errors as $\tilde{\boldsymbol{y}} := \diag(\epsilon\boldsymbol{\queryprop}/2) \boldsymbol{y},$ which are distributed as $ \tilde{\boldsymbol{y}} \sim \textup{Laplace}(\mathbf{0}, I),$ the rescaled workload as $\widetilde{W}:=\diag(\epsilon\boldsymbol{\queryprop}/2) W,$ and the rescaled response variable as $\tilde{\boldsymbol{z}} := \diag(\epsilon\boldsymbol{\queryprop}/2) \boldsymbol{z} = \widetilde{W} \boldsymbol{x} + \tilde{\boldsymbol{y}}.$ In the case of $\rho$--zCDP, we have $\boldsymbol{y} \sim \textup{N}(\mathbf{0}, \diag(\boldsymbol{1}\oslash (\rho\boldsymbol{\queryprop}))),$ so we can define the rescaled errors as $\tilde{\boldsymbol{y}} := \diag(\sqrt{\rho\boldsymbol{\queryprop}}) \boldsymbol{y},$ which are distributed as $ \tilde{\boldsymbol{y}} \sim \textup{N}(\mathbf{0}, I),$ the rescaled workload as $\widetilde{W}:=\diag(\sqrt{\rho\boldsymbol{\queryprop}}) W,$ and the rescaled response variable as $\tilde{\boldsymbol{z}} := \diag(\sqrt{\rho\boldsymbol{\queryprop}}) \boldsymbol{z} = \widetilde{W} \boldsymbol{x} + \tilde{\boldsymbol{y}}.$ 

One simple example of a matrix mechanism is a mechanism that releases the weighted least squares estimates of $W \boldsymbol{x}.$ Specifically, this can be done by first estimating $\boldsymbol{x}$ as $\boldsymbol{\hat{x}} := \argmin_{\boldsymbol{x}} \lVert  \widetilde{W} \boldsymbol{x} - \tilde{\boldsymbol{z}} \rVert^2_2= (\widetilde{W}^\top \widetilde{W})^{-1} \widetilde{W}^\top \tilde{\boldsymbol{z}},$ and then defining the output of the mechanism as $W \boldsymbol{\hat{x}}.$ Note that, in either the case of $\epsilon$--DP or $\rho$--zCDP, the privacy guarantee of the final mechanism follows from the fact that a mechanism that released $\tilde{\boldsymbol{z}}$ would satisfy the same privacy guarantee, along with the invariance to post-processing property.

This simple mechanism can be generalized by making a distinction between the linear query answers provided as output and the linear queries used to estimate $\boldsymbol{\hat{x}}.$ Specifically, let the\textit{strategy} matrix $A\in\{0,1\}^{p \times n}$ be defined as $A := \stack(\{R[i]\}_i^r),$ where each $R[i]\in \{R[i]\}_i^r$ is a query group matrix, and the query group PLB proportions as $\{\geopropStrat_i\}_{i}^r.$ We also assume that the strategy matrix has full column rank. In the same manner as described above for $W,$ we define the rescaled strategy matrix as $\widetilde{A}:=\diag(\epsilon\boldsymbol{\geopropStrat}/2) A,$ when deriving an $\epsilon$--DP mechanism, or $\widetilde{A}:=\diag(\sqrt{\rho\boldsymbol{\geopropStrat}}) A,$ when deriving a $\rho$--zCDP mechanism. Likewise, the rescaled error vector and response variable are defined as above so that $\tilde{\boldsymbol{z}} := \widetilde{A} \boldsymbol{x} + \tilde{\boldsymbol{y}}.$ The final output of this matrix mechanism is then given by this alternative estimate of the linear queries in the workload $W,$ which are $W\boldsymbol{\hat{x}} = W (\widetilde{A}^\top \widetilde{A})^{-1} \widetilde{A}^\top \tilde{\boldsymbol{z}}.$\footnote{This discussion is also straightforward to extend to the more general case in which the workloads and strategy matrices are not required to be vertically stacked query group matrices. We focus on these specific workload and strategy matrices because they are similar to the ones we consider in the rest of the paper.}

Note that the variance matrix of this output vector is given by,

$$\textup{Var}(W\boldsymbol{\hat{x}}) = W (\widetilde{A}^\top \widetilde{A})^{-1} W^\top.$$

\noindent
Past work focuses on using this variance matrix to find strategies that provide a low expected sum of squared errors, which is given by $\textup{Trace}(\textup{Var}(W\boldsymbol{\hat{x}})) = \textup{Trace}(W^\top W (\widetilde{A}^\top \widetilde{A})^{-1});$ see for example, \cite{li2010optimizing, mckenna2018optimizing}. We use an alternative approach to motivate the heuristic used to bypass geounits in Section \ref{bypassing}, after introducing how we represent the spine using matrices in the next section.

\section{Representing Linear Queries on the Spine} \label{def_spine}

In this section we describe how we represent the workload and the strategy matrices that include the linear queries for each geounit on the spine. To do so, we first consider the case in which the same query groups are used in all such geounits, and generalize this notation afterward to the case where the query matrix is dependent on the geolevel. 

Let the strategy matrix of the linear queries for each geounit be denoted by $B \in \{0,1\}^{m\times n},$ where $B=\stack(\{Q[i]\}_i^q)$ and each $Q[i]$ is a query group matrix. We assign an integer to each block geounit for the purpose of ordering the blocks. Specifically, suppose that the blocks are ordered lexicographically so that blocks in the same state are adjacent to one another, within each state, blocks in the same county are adjacent, etc. For example, in the case of the conventional geographic spine, this can be achieved by sorting the blocks by their 15 digit census GEOID, as the format of the GEOID is [2 digit state FIPS code][3 digit county FIPS code][6 digit census tract code][4 digit census block code].\footnote{Federal Information Processing Standard (FIPS) codes are fixed width codes assigned to states and counties. Each two digit state FIPS code is unique, and within each state, each three digit county FIPS code is also unique.} We refer to geounit $u \in \{1,\dots,U[l]\}$ in geolevel $l\in \{1,\dots, L\}$ as geounit $(l,u).$

Now let $b[l,u]$ denote the number of block level descendants of geounit $(l,u).$ For example, since all census blocks are descendants of the root geounit, there is a total of $b[1,1]$ blocks.\footnote{For the Commonwealth of Puerto Rico, all blocks are descendant from the PR geounit.} Let $\boldsymbol{x}\in \zz^{n b[1,1]}_+$ be a vector of histogram cell counts over all blocks, ordered in the manner described in the preceding paragraph. Using this notation, we can express the query answers to the query matrix $B$ for the root geounit in terms of these block-level cell counts as $(\mathbf{1}^\top_{b[1,1]} \otimes B) \boldsymbol{x}.$ Likewise, the query answers for all block geounits is $(I_{b[1,1]}\otimes B) \boldsymbol{x}.$ More generally, the query answers of geolevel $l$ can be expressed in terms of the block-level cell counts by $(A[l]\otimes B)\boldsymbol{x},$ where  $A[l]:=\textup{block\_diag}(\{\boldsymbol{1}_{b[l,u]}^{\top}\}_{u=1}^{U[l]})$ and $\textup{block\_diag}(\{E[1],E[2],\dots\})$ denotes a block diagonal matrix with the (possibly non-square) matrices $E[1],E[2],\dots$ along the block diagonal. We will refer to $A:=\stack(\{A[l]\}_{l=1}^L)$ as the matrix representation of the spine throughout the paper because this matrix encodes the adjacency relationships between geounits on the spine, as described in the following example. 

\begin{example}
    In this example we will describe a matrix representation of a simple spine. Specifically, suppose there are three geolevels. For geolevel $l=1,$ suppose the root (or US) geounit has two children in geolevel $l=2.$ Suppose the first geounit of geolevel $l=2$ has one child geounit in geolevel $l=L=3,$ and the second geounit in geolevel $l=2$ has two children. In this case, $A[1],A[2],A[3],$ and the matrix representation of the spine, \textit{i.e.}, the matrix $A,$ are given by 

    \begin{align*}
        A[1] &= \textup{block\_diag}(\{\mathbf{1}_{b[1,1]}^\top\}) = \mathbf{1}_3^\top = \begin{bmatrix}
            1 & 1 & 1
        \end{bmatrix}, \\
        A[2] &= \textup{block\_diag}(\{\mathbf{1}_{b[2,1]}^\top, \mathbf{1}_{b[2,2]}^\top\}) = \begin{bmatrix}
            1 & 0 & 0 \\
            0 & 1 & 1 
        \end{bmatrix}, \\
        A[3] &= \textup{block\_diag}(\{\mathbf{1}_{b[3,1]}^\top, \mathbf{1}_{b[3,2]}^\top,\mathbf{1}_{b[3,3]}^\top\}) = I_3 =  \begin{bmatrix}
            1 & 0 & 0 \\
            0 & 1 & 0 \\
            0 & 0 & 1
        \end{bmatrix}, 
    \end{align*}

    \noindent
    and
    
    \begin{align*}
        A &= \stack(\{A[1], A[2], A[3]\}) = \begin{bmatrix}
            1 & 1 & 1 \\
            1 & 0 & 0 \\
            0 & 1 & 1 \\
            1 & 0 & 0 \\
            0 & 1 & 0 \\
            0 & 0 & 1
        \end{bmatrix}.
    \end{align*}
    \qed
\end{example}

This notation allows us to express the strategy matrix containing all linear queries in all geolevels for the case in which the per-geounit strategy matrix is the same in each geolevel; specifically, the overall strategy matrix in this case is given by

$$
S := A \otimes B \in \{0,1\}^{m \sum_l U[l] \times n b[1,1]},
$$

\noindent
where, as described above, the per-geounit strategy matrix $B$ has dimension $ m \times n.$ As a summary of this notation, an example of $A \otimes B$ is shown in Figure \ref{fig:query_mat_full}. 

The notation above is not quite general enough to capture all possible strategy matrices that a user may specify for a DAS execution in the most general case. This is because the DAS also allows users to specify distinct strategy matrices in each geolevel, but the notation above assumes that the per-geounit strategy matrix in each geolevel is $B.$ For this reason, we use an alternative definition for the strategy matrix that is general enough to encompass all possible strategies used by the DAS. To do so, let $B[l] \in \{B[l]\}_{l=1}^L$ denote the per-geounit strategy matrix for geounits in geolevel $l.$ Using this notation, the full strategy matrix is,

\begin{align} \label{B_generalized}
S := \left[\begin{array}{c}
A[1] \otimes B[1]\\
A[2]\otimes B[2]\\
\vdots\\
A[L-1]\otimes B[L-1]\\
A[L]\otimes B[L]
\end{array}\right].
\end{align}

\noindent
In contrast to the strategy matrix, $S,$ we do not require the workload to consist solely of linear queries for a hierarchical set of geographic entities; instead we only place very limited restrictions on this matrix. Specifically, we use $W$ to denote the workload and only require that $W\in \rr^{ f \times  n b[1,1] }.$

When either the conventional spine or an AIAN spine is used within the DAS, the spine optimization routines are not used and the strategy matrix used by the DAS is simply the unaltered initial strategy matrix. In these cases, we use $\geoprop_{l}\in \{\geoprop_{l} \}_{l=1}^L,$ where $\geoprop_{l}\in \rr_{++}$ and $\sum_l \geoprop_{l}=1,$ to denote the proportion of the global PLB that is used for the query groups of geolevel $l.$ For each geolevel $l,$ we use $\queryprop_{i,l} \in \{\queryprop_{i,l}\}_{i=1}^{q[l]},$ where $\queryprop_{i,l}\in \rr_{++}$ and $\sum_i \queryprop_{i,l} = 1,$ to denote the proportion of the geolevel PLB that is used for query group $i.$ In summary, the proportion of the global PLB assigned to query group $i$ for geounit $(l,u)$ is $\geoprop_{l} \queryprop_{i,l}$ in these cases. 

\begin{figure}
    \centering
    \includegraphics[width=45mm]{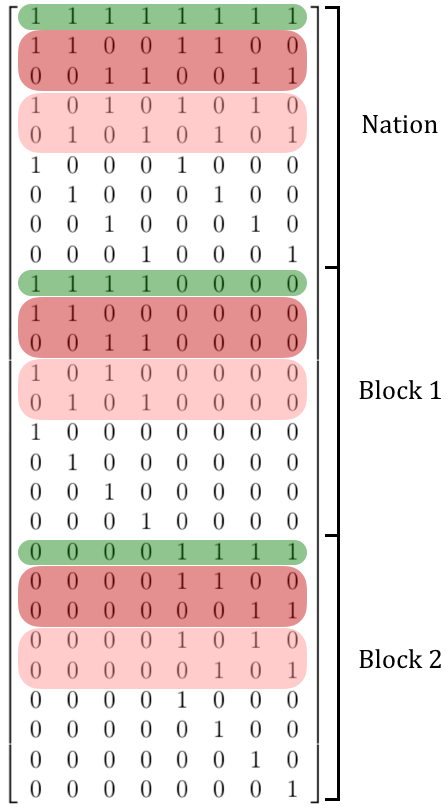}
    \caption{An example of $A \otimes B$ is given above. In this example, $A$ is given by $\stack(\mathbf{1}_{2}^{\top}, I_{2}),$ which corresponds to the case in which there is one US geounit with two block geolevel child geounits. Each geounit contains a $2\times 2$ histogram, and the linear queries in $B$ are given by a total sum query ($\mathbf{1}_2^\top\otimes \mathbf{1}^\top_2$), the marginal query groups for each attribute ($I_2 \otimes \mathbf{1}_2^\top$ and $\mathbf{1}_2^\top \otimes I_2$), and the detailed cell query group ($I_2 \otimes I_2$). For each of the geounits, the total sum query is highlighted in green, the marginal query groups are highlighted in light and dark red, and the detailed cell query group is unhighlighted.}
    \label{fig:query_mat_full}
\end{figure}

The following theorem uses the notation above to show that the DAS is either $\epsilon$--DP or $\rho$--zCDP, depending on whether a pure DP or zCDP framework and implementing mechanism is used, when either the conventional spine or an AIAN spine is used. The fact that geounits can be bypassed in the case of the optimized spine requires some modifications to this proof, so this result is provided in Section \ref{bypassing}.

\begin{theorem} \label{tda_is_dp}
Suppose the implementing mechanism $\mech(\boldsymbol{x})$ outputs $\{(A[l]\otimes B[l]) \boldsymbol{x} + \boldsymbol{y}_{l}\}_{l=1}^L,$ where each $B[l]$ is defined by vertically stacking the query group matrices $\{Q[i,l]\}_{i=1}^{q[l]}$ with $Q[i,l]$ of dimension $m[i,l]\times n,$ and $\boldsymbol{y}_{l}$ is a vector of independent random variables. Also, let $\boldsymbol{\queryprop[l]} := \stack(\{\mathbf{1}_{m[i,l]} \queryprop_{i,l} \}_i).$ Then we have the following.

\begin{enumerate}
    \item If either 
    \begin{center}
        $\boldsymbol{y}_{l} \sim \textup{Laplace}(\mathbf{0}, \diag(\mathbf{2} \oslash (\epsilon \geoprop_{l} \boldsymbol{\queryprop[l]})))$ or $\boldsymbol{y}_{l} \sim \textup{Laplace}_\zz(\mathbf{0}, \diag(\mathbf{2} \oslash (\epsilon \geoprop_{l} \boldsymbol{\queryprop[l]}))),$
    \end{center}
    \noindent
    then both $\mech(\cdot)$ and the DAS are $\epsilon$--DP with respect to the neighbor definition $\{x, x^\prime \in \dat^{\boldsymbol{d}} \cap \mathcal{G} \mid  d_{\mathcal{H}}(x, x^\prime) = 2 \}.$
    \item If either 
    \begin{center}
        $\boldsymbol{y}_{l} \sim \textup{N}(\mathbf{0}, \diag(\mathbf{1} \oslash (\rho \geoprop_{l} \boldsymbol{\queryprop[l]})))$ or $\boldsymbol{y}_{l} \sim \textup{N}_\zz(\mathbf{0}, \diag(\mathbf{1} \oslash (\rho \geoprop_{l} \boldsymbol{\queryprop[l]}))),$
    \end{center}
    \noindent
    then both $\mech(\cdot)$ and the DAS are $\rho$--zCDP with respect to the neighbor definition $\{x, x^\prime \in \dat^{\boldsymbol{d}} \cap \mathcal{G} \mid  d_{\mathcal{H}}(x, x^\prime) = 2 \}.$
\end{enumerate}
\end{theorem}
\begin{proof}
Note that for each fixed $l,$ we have

$$A[l]\otimes B[l] = A[l]\otimes \stack(\{Q[i,l]\}_i)= \stack(\{A[l]\otimes Q[i,l]\}_i).$$

\noindent
Since each $A[l]\otimes Q[i,l]$ satisfies the definition of a query group matrix, Lemma \ref{per_geostrat_is_dp} implies that releasing the output of $\mech_l,$ with output defined as $(A[l]\otimes B[l]) \boldsymbol{x} + \boldsymbol{y}_{l}$ is $(\geoprop_{l} \epsilon)$--DP in the first case and $(\geoprop_{l} \rho)$--zCDP in the second case. Thus, sequential composition implies that the DP implementing  mechanism $\mech(\cdot)$ is $\sum_l\geoprop_{l} \epsilon=\epsilon $--DP in the first case and $\sum_l\geoprop_{l} \rho=\rho $--zCDP in the second case. Since both $\rho$--zCDP and $\epsilon$--DP guarantees are invariant to post-processing, the DAS satisfies the same privacy guarantee as the implementing  mechanisms.
\end{proof}

In contrast to cases in which either the conventional spine or an AIAN spine are used, the spine optimization routines alter the initial matrix representation of the spine, $A,$ and the PLB proportions allocated to each geounit. We use $\geopropStrat_{l,u}$ to denote the proportion of the global PLB allocated to geounit $(l,u)$ after spine optimization, and $\boldsymbol{\geopropStrat}$ to denote a vector composed of these values, \textit{i.e.}, $ \boldsymbol{\geopropStrat}:= \stack(\{\geopropStrat_{l,u}\}_{l,u}).$ In other words, after spine optimization is complete, the proportion of the global PLB that is allocated to query group $i$ for geounit $(l,u)$ is given by $\geopropStrat_{l,u} \queryprop_{i,l}.$ Note that the spine optimization routines do not alter the per-geounit workloads, $ \{B[l]\}_{l},$ or the per-geounit query group proportions, $\{\queryprop_{i,l}\}_{i,l}.$ Table \ref{tab:notation} provides a summary of the notation introduced in this section that we use in the rest of the paper.

\begin{table}[h]
\centering
\begin{tabular}{ l l }
    $ U[l]\in\nn$ & The number of geounits in geolevel $l\in\{1,\dots,L\}$ \\
    $b[l,u] \in \{1,\dots, U[L]\}$ & The number of block descendants of geounit $(l,u)$ \\
    $\queryprop_{i,l} \in \rr_{++} $ & For each geounit in geolevel $l,$ the proportion of the geounit's PLB \\ 
    & allocated to query $i,$ \textit{i.e.}, $\sum_i\queryprop_{i,l}=1$ for every $l$ \\
    $ \geoprop_{l} \in \rr_{++} $ & The proportion of the global PLB that is allocated to each geounit in \\ 
    & geolevel $l$ before spine optimization  \\
    $\geopropStrat_{l,u}\in \rr_+ $ & The global PLB proportion of geounit $(l,u)$ after spine optimization \\
    $ B[l]\in \{0,1\}^{m \times n}$ & The per-geounit strategy matrix for geolevel $l,$ defined by stacking \\
    & the query group matrices for geolevel $l$ \\
    $A[l]\in \{0,1\}^{U[l]\times b[1,1]}$ & The matrix representation of geolevel $l$ of the spine, \textit{i.e.}, \\
    & $A[l]=\textup{block\_diag}(\{\boldsymbol{1}_{b[l,u]}^{\top}\}_{u})$ \\
    $A\in \{0,1\}^{\sum_l U[l] \times b[1,1]}$ & The matrix representation of the spine, \textit{i.e.}, $A=\stack(\{A[l]\}_{l=1}^L)$ \\ 
    $S\in \{0,1\}^{m \sum_l U[l] \times n b[1,1]}$ & The strategy matrix, as defined in (\ref{B_generalized})
\end{tabular}
\caption{Definitions used to construct the strategy matrices used by the DAS} \label{tab:notation}
\end{table}

\section{Bringing Off-Spine Entities Closer to the Spine} \label{osed}

This section describes two methods supported by the spine optimization routines to enhance accuracy in a set of target OSEs by bringing these geographic units closer to the spine in the first stage of the spine optimization routines. Both of these methods sacrifice accuracy in one or more geographic levels of the input spine to achieve this goal because they replace one or more geographic levels of the input geographic spine with alternative geographic levels that are closer to the OSEs. For example, the 2020 production redistricting and DHC spine settings resulted in the block group geographic level in an AIAN input spine being replaced with optimized block group geolevels, which degraded accuracy for census block groups. 

First, the approach used by the 2020 production redistricting and DHC DAS implementations to bring target OSEs closer to the geographic spine is conceptually the simplest approach. These DAS implementations define this first stage of the spine optimization settings by directly defining geographic levels as the intersections of geographic entities. For example, each optimized block group in the optimized spine used for the 2020 redistricting data file production execution is defined as an intersection of a combination of OSEs and geounits, \textit{e.g.}, a tract, an AIAN area (or the region outside of all AIAN areas), \textit{etc.}; see Section \ref{settings} for a full description of this approach.

The next subsection describes the second alternative approach that is supported by the DAS but that is not used in our 2020 production DAS execution settings. Specifically, rather than  directly defining the new geographic levels in the optimized spine as intersections of OSEs, this alternative method automates the geounit definitions in the new geographic levels to bring OSEs closer to the optimized spine.

\subsection{Automating the Choice of Alternative Geolevel Definitions}

This section describes one approach to redefine geolevels to bring OSEs closer to the spine. This is done by using a heuristic to reduce an objective function based on the \textit{off-spine entity distance} (OSED), which is the minimum number of geounits that must be added or subtracted from one another to derive an OSE.\footnote{Our use of the term ``off-spine entities" refers to geographic entities that may be off of the DAS geographic spine, but we do not assume that each OSE has a geographic extent that differs from each (on-spine) geounit. An implication of this definition of OSED is that, when the geographic extent of an OSE is identical to that of a geounit on the spine, its OSED is equal to one.} This heuristic generally results in a reduction in the variance of estimates for these geographic regions, particularly when a sufficiently high proportion of the global PLB is allocated to each geolevel. 

Before defining a systematic algorithm that outputs the OSEDs for each OSE, we define some additional notation. Let the set of OSEs be denoted by $\mathcal{K},$ and let $\{C_k(u)\}_{k\in\mathcal{K}},$ where $C_k:\nn \rightarrow \{0,1\},$ denote a set of functions such that $C_k(u)$ is equal to one when the OSE $k\in \mathcal{K}$ contains the block geounit $u$ and is equal to zero otherwise. We assume the input and output spines of this stage of the spine optimization routines are encoded using the matrix representation of the spine, as defined in the previous section and denoted by $A.$ The algorithm we describe for computing OSED next requires computing intermediate OSEDs of intersections of a given geographic entity and geounits on the spine; we denote the OSED of the intersection of the geographic entity $k\in\mathcal{K}$ and the geounit $(u,l)$ as $c[k,l,u].$

We begin by fixing some $k \in \mathcal{K}$ and finding the OSEDs for both $k$ and its complement, which we denote by $k',$ under the (temporary) assumption that the only geolevel is the block geolevel. In this case, each block would contribute $c[k,L,u] := C_k(u)$ to the OSED of OSE $k.$ Likewise, each block would contribute $c[k',L,u] := 1- C_k(u)$ to the OSED of the complement of $k.$

If the block-group geolevel were to be added to the spine at this point, we could compute the intersection of a single block group and $k$ in one of two ways. First, we could add all the blocks together that are both inside of entity $k$ and inside of the block group. This would result in the block-group $u$ contributing $\sum_{v\in \textup{Children}(u)} c[k,L,v]$ to the OSED of $k.$ On the other hand, we could also take all the geographic extent that the block group occupies and then subtract off the geographic region of the blocks in the complement of $k.$ This would result in this block group contributing $1 + \sum_{v\in \textup{Children}(u)} c[k',L,v]$ to the OSED of $k.$ Note that an additional one is added in this case because of the additional step of subtracting the complement of $k$ in the block group from the block group itself. Since OSED is defined as the \textit{minimum} number of geounits that must be added or subtracted to one another to define an entity, we choose the option that results in a smaller value. Since similar derivations can be carried out for the complement of $k,$ by symmetry, we have,

\begin{align} \label{recur_pre1}
    c[k,L-1,u] &:= \min \left\{\sum_{v\in \textup{Children}(u)} c[k,L,v], \; 1 + \sum_{v\in \textup{Children}(u)} c[k',L,v]  \right\} \\
    c[k',L-1,u] &:=  \min \left\{\sum_{v\in \textup{Children}(u)} c[k',L,v],\; 1 + \sum_{v\in \textup{Children}(u)} c[k,L,v]  \right\}. \label{recur_pre2}
\end{align}

Similar logic can be repeated to derive to derive the following recursive system of equations for the OSED of the intersection of entity $k$ and an arbitrary geounit $(l-1,u)$
\begin{align} \label{recur1}
    c[k,l-1,u] &:= \min \left\{\sum_{v\in \textup{Children}(u)} c[k,l, v], \; 1 + \sum_{v\in \textup{Children}(u)} c[k',l,v]  \right\} \\
    c[k',l-1,u] &:=  \min \left\{\sum_{v\in \textup{Children}(u)} c[k',l,v],\; 1 + \sum_{v\in \textup{Children}(u)} c[k,l,v]  \right\} \label{recur2}
\end{align}
\noindent Since all entities are assumed to be contained within the US, the final OSED for entity $k$ can be found by applying these recursions up to the root geounit and defining this OSED as $c[k,1,1].$ Afterward, our final objective function is defined by applying the reduce operation $h:\nn^{\textup{Card}(\mathcal{K})} \rightarrow \rr$ to the OSEDs, which, for example, can be defined as the arithmetic mean or the max function. This computation is summarized in Algorithm \ref{OSED}.

\LinesNotNumbered
\begin{algorithm}[H] \label{OSED}
\SetArgSty{textrm}
\DontPrintSemicolon
\caption{OSEDs\_Reduced$(A, \mathcal{K},\{C_k(\cdot)\}_{k}, h(\cdot))$} 
\For{$k \in \mathcal{K}$}{
\For{geounit $u$ in geolevel $L$}{   
$c[k,L,u] \leftarrow C_k(u)$ \\ 
$c[k',L,u] \leftarrow 1 - C_k(u)$ \\
}  
\For{ geolevel $l$ in $\{L-1,L-2,\dots,1\}$}{
\For{geounit $u$ in geolevel $l$}{
$x \leftarrow \sum_{v \in \textup{Children}(u)} c[k,l+1,v]$ \\ 
$y \leftarrow \sum_{v \in \textup{Children}(u)} c[k',l+1,v]$ \\ 
$c[k,l,u] \leftarrow \min \; \{ x, y + 1\} $ \\ 
$c[k',l,u] \leftarrow  \min \; \{ y, x + 1\}$ 
}
}
}
$\mathbf{return}$ $h(\{c[k,1,1]\}_{k\in\mathcal{K}})$
\end{algorithm}

In the most general setting, in which the OSEs are not necessarily disjoint, formulating an algorithm that redefines certain geolevels in order to minimize the OSEDs with a polynomial time complexity appears to be a difficult problem because of the similarity of this optimization problem to a set covering problem. For this reason, Algorithm \ref{opt_gs} describes an example of a greedy approach to approximate the redefinition of tract groups that minimizes the OSEDs of a set of off-spine entities in a computationally tractable manner. This example first redefines block groups as optimized block groups by combining blocks within a given tract geounit, and within the same intersection of the OSEs. As described in the pseudocode below, each optimized block group is composed of up to $\sqrt{n} + \textup{fanout\_cutoff}$ blocks, where $n$ is the number of blocks in the optimized block group's parent tract and fanout\_cutoff is a user choice parameter. The goal of this choice is to ensure the fanout value, \textit{i.e.}, the number of child geounits of a given geounit, of the tract and the fanout values of its child optimized block groups are reasonably low. This is beneficial because high fanout values can increase the runtime of DAS executions (by increasing the size of the optimization problems used within the DAS), and we have found that high fanout values also appear to degrade accuracy. To see why using this upper bound for the number of blocks within each block group results in reasonably low fanout values, note that we can minimize the  highest fanout value among a given tract and its children by defining the optimized block groups so that they each contain at most $ \left \lceil{\sqrt{n}}\right \rceil,$ where $\left \lceil{\cdot}\right \rceil$ denotes the ceiling function. For example, for a tract with 100 block geounit descendants, redefining its children by 10 optimized block group child geounits, each with 10 block child geounits, would result in the lowest possible maximum fanout of 10 among these 11 geounits.  After the optimized block group geolevel has been defined in Algorithm \ref{opt_gs}, the tract group geolevel is redefined to reduce the OSEDs of the entities in $\mathcal{K}.$ Note that this stage of the algorithm also ensures that the maximum number of tracts within each tract group is less than $\sqrt{n} + \text{fanout\_cutoff}.$ 

Algorithm \ref{opt_gs}  uses $\boldsymbol{x} \preceq_{\textup{lexicographic}} \boldsymbol{y},$ where $\boldsymbol{x},\boldsymbol{y} \in \rr^{n},$ to denote the lexicographic less than or equal to partial ordering, which is defined as $\bot$ when $\boldsymbol{x}\neq \boldsymbol{y}$ and the first index $i$ for which $\boldsymbol{x}[i] $ and $\boldsymbol{y}[i]$ differ satisfies $\boldsymbol{x}[i]>\boldsymbol{y}[i],$ and $\top$ otherwise.

Note that Algorithm \ref{opt_gs} does not alter the PLB proportions of the input spine. The next section describes the algorithm used to update these proportions in the second spine optimization stage.

\begin{algorithm}\label{opt_gs}
\SetArgSty{textrm}
\DontPrintSemicolon
\caption{Redefine\_Block\_Groups\_and\_Tract\_Groups$(A, \mathcal{K},\{C_k(\cdot)\}_{k})$} 
\tcp{Within each tract, redefine block groups by combining groups of $\sqrt{n} +\textup{fanout\_cutoff}$ blocks in the intersections of the same OSEs, where $n$ is the number of blocks in the tract.} 
\tcp{Initialize tract groups so that they each have one child Tract.} 
current\_OSEDs $\gets \textup{OSEDs\_Reduced}(A, \mathcal{K},\{C_k(\cdot)\}_{k}, \textup{Sort\_Descending}(\cdot))$ \\
\For{$i \in \{1,\dots\}$}{
$\textup{altered\_spine} \gets \bot$ \\
\For{$\textup{county} \in \textup{Counties}$}{  
$n\gets$ Number\_of\_tracts\_in\_county$($county$)$ \\
\For{$u,v \in \textup{Children}(\textup{county})$}{
\If{$\textup{Card}(\textup{Children}(u) \cup \textup{Children}(v)) >\sqrt{n} + \textup{fanout\_cutoff}$}{
\textbf{continue}\\
} 
\tcp{Combine\_Siblings$(u,v,A)$ returns a new spine defined as the spine $A$ after replacing the sibling geounits $u$ and $v$ with a single geounit with children given by $\textup{Children}(u)\cup\textup{Children}(v):$}
$A^\prime \gets $ Combine\_Siblings$(u,v,A)$ \\
test\_OSEDs $\gets \textup{OSEDs\_Reduced}(A^\prime, \mathcal{K},\{C_k(\cdot)\}_{k}, \textup{Sort\_Descending}(\cdot))$ \\
\If{test\_OSEDs $\preceq_{\textup{lexicographic}} $ current\_OSEDs}{ 
$A \gets A^\prime$ \\
current\_OSEDs $ \gets $ test\_OSEDs \\
$\textup{altered\_spine} \gets\top$\\
}
}
}
\If{$ \textbf{not } \textup{altered\_spine}$}{
\textbf{break}\\
}
}
$\mathbf{return}$ $A$ 
\end{algorithm}

\section{A Pareto Frontier of Geounit Definitions} \label{bypassing}

This section uses a matrix mechanism to derive decision rules for whether or not to bypass a geounit. We will consider the case in which a pure-DP implementing mechanism is used within the DAS first. This decision rule is motivated by a setting that is less general than that of the DAS in several ways. First, for every $l,l'\in\{1,\dots, L\},$ we suppose $B[l]=B[l'],$ and let $B:=B[l].$ Second, for each query group $i,$ we suppose $ \queryprop_{i,l} = \queryprop_{i,l'},$ and let $\queryprop_i := \queryprop_{i,l}.$ Third, we also constrain our attention to DP implementing mechanisms that use noise drawn from the continuous distributions $\textup{Laplace}(\cdot)$ rather than its discrete counterpart $\textup{Laplace}_{\zz}(0, b).$ The decision rule is based on the variance matrix of the WLS estimator of the detailed cell histogram counts for all block geounits in this simplified setting. It will be helpful to derive this variance prior to describing the decision rule. To do so, recall that the number of rows of query group $i$ is denoted by $m[i,1],$ and let $\boldsymbol{\queryprop} := \stack(\{\mathbf{1}_{m[i,1]} \queryprop_i \}_i).$  In the terminology introduced in Section \ref{matmech}, let the rescaling vector $r\in\rr^{m \sum_l U[l]}_{+}$ in this case be defined as $r:=\boldsymbol{\geopropStrat} \otimes \boldsymbol{\queryprop}\epsilon$ and let the rescaled strategy matrix be defined as

$$
\widetilde{S} := \diag(r) S = (\diag(\boldsymbol{\geopropStrat}) \otimes \diag(\boldsymbol{\queryprop}\epsilon) ) (A\otimes B) = (\diag(\boldsymbol{\geopropStrat}) A) \otimes (\diag(\boldsymbol{\queryprop}\epsilon) B).
$$

\noindent
Also, let the stacked DP implementing  mechanism answers be defined as $\tilde{\boldsymbol{z}}:= \widetilde{S} \boldsymbol{x} + \tilde{\boldsymbol{y}},$ where $\tilde{\boldsymbol{y}} = \diag(r)\boldsymbol{y},$ $\boldsymbol{y}\sim \textup{Laplace}(\mathbf{0}, \diag(\mathbf{2} \oslash (\epsilon \boldsymbol{\geopropStrat} \otimes \boldsymbol{\queryprop}))).$

Then the WLS estimate for this strategy matrix can be expressed as

\begin{align}
    \boldsymbol{\hat{x}} = (\widetilde{S}^\top \widetilde{S})^{-1} \widetilde{S} \tilde{\boldsymbol{z}}, 
\end{align}

\noindent
and the variance matrix of the output of the matrix mechanism, $\textup{Var}(W\boldsymbol{\hat{x}}),$ is proportional to

\begin{align} \label{var_mat}
    W (\widetilde{S}^\top \widetilde{S})^{-1} W^\top.
\end{align}

It is also worth explicitly stating how we define the operation of bypassing a parent geounit. We define the operation of bypassing a parent geounit with $\nc$ children, each with equal PLB proportions, by, 1) creating $\nc$ geounits in the geolevel of the parent, each with a geolevel proportion given by the sum of the proportions allocated to the parent and one of the children, 2) defining the single child of each of these $\nc$ new geounits by one of the children, 3) removing the old parent geounit, 4) redefining the geolevel proportion of each child to be zero. Thus, even though we call this operation ``bypassing a parent," this operation actually moves the geolevel PLB proportion to a higher geolevel. This ensures that this operation does not change the total number of geolevels, and since the DAS fixes the final estimates in a top-down manner, so that the consistency with parent constraints are satisfied, this also ensures that the entire share of the geolevel PLB allocations are used in cases in which a parent geounit has only one child geounit. Although this definition describes the operation used within the DAS, the decision rules developed in this section only depend on properties of the matrix $\widetilde{S}^\top \widetilde{S},$ and using this definition of the bypass operation impacts this matrix in the same way as simply reallocating the PLB of the parent geounit to the child geounits. For this reason, we motivate the decision rules described in this section on this simpler definition of the bypass operation. In other words, unlike what is done in the DAS codebase, here we define the operation of bypassing a parent geounit as redefining the PLB of each of the child geounits as the sum of the child's PLB and the PLB of the parent and then redefining the PLB of the parent geounit to be zero. 

The next result describes a case in which each of the expected squared errors of the matrix mechanism, \textit{i.e.}, the diagonal of $\textup{Var}(W \boldsymbol{\hat{x}}),$ can be decreased, or remain unchanged, by bypassing a parent. In cases in which the initial PLB proportion of a parent and its children are equal, this result implies that accuracy can be improved by bypassing the parent geounit when it has less than or equal to three child geounits. Note that the decision rule in the following theorem is only used in the spine optimization routines of the DAS when discrete or continuous $\epsilon$--DP Laplace mechanisms are used for the implementing  mechanisms. The decision rule used for cases in which either discrete or continuous $\rho$--zCDP Gaussian mechanisms are used for the implementing  mechanisms within the DAS, \textit{i.e.}, the decision rule used for the redistricting data file and DHC production executions, is described after the proof.

\begin{theorem} \label{pure_dp_bypassing}
Suppose that $\epsilon$--DP is implemented using Laplace mechanisms and that the per-geolevel query strategies and query PLB proportions are the same in each geolevel. Also, suppose that the geolevel PLB allocated to each of the $\nc$ children of geounit $(l,1)$ are equal (\textit{i.e.}: $\geopropStrat_{l+1,u}=\geopropStrat_{l+1,v}$ for all $u,v\in \{1,\dots, \nc\}$). If $\geopropStrat_{l+1,1} \geq  (\nc-1) \geopropStrat_{l,1}/2,$ then, for any  $B$ and $\{\queryprop_i\}_{i},$ reallocating the PLB assigned to geounit $(l,1)$ to its children will either decrease or leave unchanged each of the diagonal elements of $\textup{Var}(W \boldsymbol{\hat{x}}).$ 
\end{theorem}
\begin{proof}
Let $\{\geopropStrat_0[l,u]\}_{l,u}$ and $\{\geopropStrat_1[l,u]\}_{l,u}$ denote the geounit PLB proportions before and after reallocation, respectively. Likewise, let $\boldsymbol{\hat{x}}_0$ and $\boldsymbol{\hat{x}}_1$ denote the WLS estimate before and after reallocation, respectively. Also, for the symmetric matrices $C,D\in \rr^{n\times n}$ we will use $C \leq D$ to denote the condition that $C-D$ is negative semidefinite and $C \leq 0$ to denote the condition that $C$ is negative semidefinite. The variance matrix of the WLS estimate  before (respectively, after) bypassing geounit $(l,1)$ is proportional to

\begin{align*}
\textup{Var}(\boldsymbol{\hat{x}}_i) = &
((\textup{diag}(\boldsymbol{\geopropStrat}_i) A \otimes \diag(\boldsymbol{\queryprop}) B)^\top (\textup{diag}(\boldsymbol{\geopropStrat}_i) A \otimes \diag(\boldsymbol{\queryprop}) B))^{-1} \\ = &  ((A^\top \textup{diag}(\boldsymbol{\geopropStrat}_i) \otimes B^\top \diag(\boldsymbol{\queryprop}) ) (\textup{diag}(\boldsymbol{\geopropStrat}_i) A \otimes \diag(\boldsymbol{\queryprop}) B))^{-1} \\
= & ((A^\top \textup{diag}(\boldsymbol{\geopropStrat}_i^2) A) \otimes (B^\top \diag(\boldsymbol{\queryprop}^2) B))^{-1} \\
= & (A^\top \textup{diag}(\boldsymbol{\geopropStrat}_i^2) A)^{-1} \otimes (B^\top \diag(\boldsymbol{\queryprop}^2)B)^{-1},
\end{align*}

\noindent where $i=0$ (respectively, $i=1$). We prove the sufficient condition that $\textup{Var}(\boldsymbol{\hat{x}}_1) \leq \textup{Var}(\boldsymbol{\hat{x}}_0).$ Given the variance matrix of the WLS estimator above, this condition is equivalent to

\begin{align*}
\textup{Var}(\boldsymbol{\hat{x}}_1) \leq \textup{Var}(\boldsymbol{\hat{x}}_0) \iff \\
(A^\top \textup{diag}(\boldsymbol{\geopropStrat}_1^2) A)^{-1} \otimes (B^\top \diag(\boldsymbol{\queryprop}^2) B) ^{-1} \leq (A^\top \textup{diag}(\boldsymbol{\geopropStrat}_0^2) A)^{-1} \otimes  (B^\top \diag(\boldsymbol{\queryprop}^2) B) ^{-1}  \iff \\
((A^\top \textup{diag}(\boldsymbol{\geopropStrat}_1^2) A)^{-1} - (A^\top \textup{diag}(\boldsymbol{\geopropStrat}_0^2) A)^{-1}) \otimes (B^\top \diag(\boldsymbol{\queryprop}^2) B) ^{-1} \leq 0.
\end{align*}

\noindent Since we assume that $B^\top \diag(\boldsymbol{\queryprop}^2) B$ is positive definite, this condition holds if and only if,

\begin{align*}
((A^\top \textup{diag}(\boldsymbol{\geopropStrat}_1^2) A)^{-1} - (A^\top \textup{diag}(\boldsymbol{\geopropStrat}_0^2) A)^{-1}) \leq 0 \iff \\
(A^\top \textup{diag}(\boldsymbol{\geopropStrat}_1^2) A)^{-1} \leq (A^\top \textup{diag}(\boldsymbol{\geopropStrat}_0^2) A)^{-1}
\iff \\
A^\top \textup{diag}(\boldsymbol{\geopropStrat}_1^2) A \geq A^\top \textup{diag}(\boldsymbol{\geopropStrat}_0^2) A \iff \\
A^\top \textup{diag}(\boldsymbol{\geopropStrat}_1^2 - \boldsymbol{\geopropStrat}_0^2) A \geq 0
\end{align*}

Note that the only elements of $\boldsymbol{\geopropStrat}_1$ that are not equal to $\boldsymbol{\geopropStrat}_0$ can be defined in terms of $\boldsymbol{\geopropStrat}_0$ by $\geopropStrat_1[l+1,u] = \geopropStrat_0[l+1,u] + \geopropStrat_0[l,1]$ for all $u \in \{1,\dots \nc\}$ and $\geopropStrat_1[l,1]=0.$ Thus, $ A^\top \textup{diag}(\boldsymbol{\geopropStrat}_0^2-\boldsymbol{\geopropStrat}_1^2) A $ is a block matrix with a single block that is nonzero. Let $n[u]:=b[l+1,u]$ for $u\in \{1\dots,\nc\}.$ This block is of dimension $b[l,1]\times b[l,1],$ and is given by 

$$D := \geopropStrat_0[l,1] \left[\begin{array}{cccc}
2\geopropStrat_0[l+1,1]\mathbf{1}_{n[1]} \mathbf{1}_{n[1]}^\top & -\geopropStrat_0[l,1]\mathbf{1}_{n[1]} \mathbf{1}_{n[2]}^\top & -\geopropStrat_0[l,1]\mathbf{1}_{n[1]} \mathbf{1}_{n[3]}^\top & \cdots \\
-\geopropStrat_0[l,1]\mathbf{1}_{n[2]} \mathbf{1}_{n[1]}^\top & 2\geopropStrat_0[l+1,2]\mathbf{1}_{n[2]} \mathbf{1}_{n[2]}^\top & -\geopropStrat_0[l,1]\mathbf{1}_{n[2]} \mathbf{1}_{n[3]}^\top & \\
-\geopropStrat_0[l,1]\mathbf{1}_{n[3]} \mathbf{1}_{n[1]}^\top & -\geopropStrat_0[l,1]\mathbf{1}_{n[3]} \mathbf{1}_{n[2]}^\top & 2\geopropStrat_0[l+1,3]\mathbf{1}_{n[3]} \mathbf{1}_{n[3]}^\top\\
\vdots &  & & \ddots \\
\end{array}
\right].$$ 

\noindent For each $i\in \{1,\dots,\nc\},$ let $a[i]:=\sum_{j=1}^{i-1} n[j],$ $b[i]:=\sum_{j=i+1}^{\nc} n[j],$ and $\boldsymbol{t_i}:=(\mathbf{0}_{a[i]}^\top,\mathbf{1}_{n[i]}^\top/n[i], $ $\mathbf{0}_{b[i]}^\top)^\top.$ Note that the matrix $D$ has only $\nc$ unique columns, and the vectors $\{\boldsymbol{t_i}\}_{i=1}^\nc$ provide an orthogonal basis for the span of these columns. Let $T:=\textup{horizontal\_stack}(\{\boldsymbol{t_i}\}_{i}).$ Thus, 

\begin{center}
    $D \geq 0 \iff TDT^\top \geq 0 \iff$
\end{center}

\begin{center}
$ \left[\begin{array}{cccc}
2\geopropStrat_0[l+1,1]\geopropStrat_0[l,1] & -\geopropStrat_0[l,1]^2 & -\geopropStrat_0[l,1]^2 & \cdots \\
-\geopropStrat_0[l,1]^2 & 2p_0(l+1,2)\geopropStrat_0[l,1]  &-\geopropStrat_0[l,1]^2   & \\
-\geopropStrat_0[l,1]^2 &  -\geopropStrat_0[l,1]^2 & 2\geopropStrat_0[l+1,3]\geopropStrat_0[l,1]  \\
\vdots &  & & \ddots \\
\end{array}
\right] \geq  0 \iff $
\par\end{center}

\begin{center}
    $D' := \diag(\mathbf{1}_r (2\geopropStrat_0[l+1,1]\geopropStrat_0[l,1]+\geopropStrat_0[l,1]^2)) - \mathbf{1}_r \mathbf{1}_r^\top \geopropStrat_0[l,1]^2  \geq 0.$
\end{center}

\noindent Since $\geopropStrat_0[l,u]\geq 0$ for all $l$ and $u,$ the eigenvector of $D'$ corresponding to the smallest eigenvalue is $\mathbf{1}_r,$ so this matrix is positive semidefinite when

\begin{center}
$2\geopropStrat_0[l+1,1]\geopropStrat_0[l,1] \geq  (\nc-1) \geopropStrat_0[l,1]^2.$
\end{center}

\noindent Since $(l,1)$ was not bypassed in our initial PLB allocation, $\geopropStrat_0[l,1] \neq 0,$ so we have

\begin{center}
$\geopropStrat_0[l+1,1] \geq  (\nc-1) \geopropStrat_0[l,1]/2.$
\end{center}

\end{proof}

\subsubsection*{Remark 1}
Consider two cases in which the only query in $B$ is a total population query. First, if the parent is not bypassed, we could construct unbiased estimates of the total population of the parent by either observing the DP answer for the parent directly, which has a variance of $2/\geopropStrat_{l,1}^2,$ or by summing the DP answers of the children together, which has a variance of $2 \nc/\geopropStrat_{l+1,1}^2.$ The mean with inverse-variance weighting provides the linear combination of these estimates with the lowest possible variance of $2 \nc/(\nc \geopropStrat_{l,1}^2 + \geopropStrat_{l+1,1}^2)$ in this case. Second, if we did bypass the parent, we could estimate the total population query for the parent by summing the total population of the children, which would have a variance of $2 \nc/(\geopropStrat_{l+1,1} + \geopropStrat_{l,1})^2$ in this case. The theorem above can be viewed as a statement that limiting our attention to a single total population query, and to only the parent and its children, in this way is without loss of generality, at least for the purpose of finding the cases in which bypassing the parent does not increase the expected error of the WLS estimator for any query in any geounit. This is because $2 \nc/(\nc \geopropStrat_{l,1}^2 + \geopropStrat_{l+1,1}^2) \geq 2 \nc/(\geopropStrat_{l+1,1} + \geopropStrat_{l,1})^2$ if and only if $\geopropStrat_{l+1,1} \geq  (\nc-1) \geopropStrat_{l,1}/2,$ which is the same requirement given in the statement of the theorem. \qed

\subsubsection*{Remark 2}
This theorem may be of independent interest in the DP literature because other authors have already considered strategy matrices that have a hierarchical structure that is analogous to the matrix representation of the spine $A$ above, and this result can be used to narrow the search space of the set of hierarchical strategy matrices considered when choosing a strategy matrix with this property \cite{hay2010boosting,li2010optimizing}. For example, when using the WLS approach described by \cite{hay2010boosting} with PLB proportions that are the same for each level of the hierarchy and with the children of each non-leaf node defined by either two or three sub-intervals of the range of the parent node, this result implies that the expected squared error of an arbitrary linear query can be reduced by increasing the number of sub-intervals used to define these child nodes. \qed

One can use a similar technique to show that, in the case in which $\rho$--zCDP implementing Gaussian mechanisms are used to define $\boldsymbol{y},$ bypassing a parent will increase at least one diagonal element of $\textup{Var}(W \boldsymbol{\hat{x}})$ whenever the parent has two or more children, and the diagonal elements of $\textup{Var}(W \boldsymbol{\hat{x}})$ will remain unchanged whenever the parent has only one child. In part for this reason, our decision rule for the case in which $\rho$--zCDP discrete or continuous Gaussian mechanisms are used within the DAS consists of only bypassing parent geounits with only one child. As described in the introduction, this decision rule can also be motivated more directly than what was done in Theorem \ref{pure_dp_bypassing}. Specifically, given the top-down manner in which the DAS fixes estimates to ensure consistency with the estimates of the parent geounits, the estimates of a child geounit are fixed by those of the parent geounit whenever the parent only has one child geounit. Thus, if we did not bypass parent geounits with only one child, the PLB allocation of these child geounits would simply not be used, in the sense that the noisy measurements of these child geounits would not impact the histogram estimates of the child geounits, as is also described in the following observation. Note that the decision rule provided in Theorem \ref{pure_dp_bypassing} results in at least the same number of geounits being bypassed, since, whenever $\nc=1,$ the decision rule provided in Theorem \ref{pure_dp_bypassing} results in bypassing being chosen for any PLB allocations of the parent and child geounits.

\begin{observation} 
The noisy answers for each child geounit that does not have a sibling geounit do not impact the output of the DAS. Thus, the variance of the noisy answers that are used to construct the histogram estimates of these child geounits can be decreased by bypassing all parent geounits with only one child.
\end{observation}

Algorithm \ref{detUpdate} summarizes how both decision rules are used within the spine optimization routines of the DAS. Note that this algorithm starts at the block group geolevel and iterates up the spine to the US geolevel, rather than starting at the US geolevel and moving downward, which ensures that no remaining geounits can be bypassed after only one pass through the spine.

\LinesNotNumbered
\begin{algorithm}[H] \label{detUpdate}
\SetArgSty{textrm}
\DontPrintSemicolon
\caption{Move\_Spine\_to\_Pareto\_Frontier$(A, \{\geoprop_{l}\}_{l}, \textup{Pure\_DP}\in\{\top, \bot\})$}
\For{$l \in \{1,\dots,L\}$}{
\For{$u$ in geolevel $l$}{
$\geopropStrat_{l,u} \gets \geoprop_{l}$
}
}
\For{$l \in \{L-1,L-2,\dots, 1\}$}{
\For{$u$ in geolevel $l$}{
Min\_Child\_PLB $ \gets \min_{c\in \textup{Children}(u)} \; \geopropStrat_{l+1,c}$ \\
$\nc \gets \textup{Card}(\textup{Children}(u))$ \\
\If{\textup{Pure\_DP} $\textup{and } \textup{Min\_Child\_PLB}  \geq (\nc-1)\geopropStrat_{l,u}/2$}{
$A, \{\geopropStrat_{l,u}\}_{l,u} \gets \textup{Bypass\_Parent\_Geounit}(u, A, \{\geopropStrat_{l,u}\}_{l,u})$ 
}
\uElseIf{$\nc = 1$}{
$A, \{\geopropStrat_{l,u}\}_{l,u}\gets \textup{Bypass\_Parent\_Geounit}(u, A, \{\geopropStrat_{l,u}\}_{l,u})$}
}
}
$\mathbf{return}$ $A, \{\geopropStrat_{l,u}\}_{l,u}$
\end{algorithm}

\subsection{The impact of Spine Optimization on the Privacy Guarantees of the DAS}

Using the DAS with the spine that is output from Algorithm \ref{opt_gs} does not alter the privacy guarantees of the DAS because the same arguments used in the proof of Theorem \ref{tda_is_dp} apply to this case as well. However, the same cannot be said for the spine that is output from Algorithm \ref{detUpdate} when at least one geounit is bypassed. The following theorem generalizes the argument used in Theorem \ref{tda_is_dp} to show that using the spine output from Algorithm \ref{detUpdate} within the DAS does not impact the privacy guarantees.

\begin{theorem} \label{tda_w_opt_spine_is_dp}
Suppose the implementing mechanism $\mech(\boldsymbol{x})$ outputs $\{(A[l][u,\cdot] \otimes B[l]) \boldsymbol{x} + \boldsymbol{y}_{l,u}\}_{l,u},$ where each $A[l]$ is the matrix representation of geolevel $l\in\{1,\dots, L\}$ of the spine output from the spine optimization routines described in Sections \ref{osed} and \ref{bypassing}, each $B[l]$ is defined by vertically stacking the query group matrices $\{Q[i,l]\}_{i=1}^{q[l]}$ with $Q[i,l]$ of dimension $m[i,l]\times n,$ $\boldsymbol{x}$ is the confidential histogram cell counts for dataset $x\in \dat^{\boldsymbol{d}} \cap \mathcal{G},$ and $\boldsymbol{y}_{l,u}$ is a vector of either independent random variables or has all elements equal to $\infty.$ Also, let $\boldsymbol{\queryprop[l]} := \stack(\{\boldsymbol{1}_{m[i,l]} \queryprop_{i,l} \}_i).$ Then we have the following.

\begin{enumerate}
    \item If $\boldsymbol{y}_{l,u}=\boldsymbol{\infty}$ when $\geopropStrat_{l,u}=0,$ and is distributed as either $\boldsymbol{y}_{l,u} \sim \textup{Laplace}(\mathbf{0}, \diag(\mathbf{2} \oslash (\epsilon \geopropStrat_{l,u} \boldsymbol{\queryprop[l]})))$ or $\boldsymbol{y}_{l,u} \sim \textup{Laplace}_\zz(\mathbf{0}, \diag(\mathbf{2} \oslash (\epsilon \geopropStrat_{l,u} \boldsymbol{\queryprop[l]})))$ when $\geopropStrat_{l,u}>0,$ then both $\mech(\cdot)$ and the DAS are $\epsilon$--DP with respect to the neighbor definition $\{x, x^\prime \in \dat^{\boldsymbol{d}} \cap \mathcal{G} \mid  d_{\mathcal{H}}(x, x^\prime) = 2 \}.$
    \item If $\boldsymbol{y}_{l,u}=\infty$ when $\geopropStrat_{l,u}=0,$ and is distributed as either $\boldsymbol{y}_{l,u} \sim \textup{N}(\mathbf{0}, \diag(\mathbf{1} \oslash (\rho \geopropStrat_{l,u} \boldsymbol{\queryprop[l]})))$ or $\boldsymbol{y}_{l,u} \sim \textup{N}_\zz(\mathbf{0}, \diag(\mathbf{1} \oslash (\rho \geopropStrat_{l,u} \boldsymbol{\queryprop[l]})))$ when $\geopropStrat_{l,u}>0,$ then both $\mech(\cdot)$ and the DAS are $\rho$--zCDP with respect to the neighbor definition $\{x, x^\prime \in \dat^{\boldsymbol{d}} \cap \mathcal{G} \mid  d_{\mathcal{H}}(x, x^\prime) = 2 \}.$
\end{enumerate}
\end{theorem}
\begin{proof}

For $x \in \dat^{\boldsymbol{d}} \cap \mathcal{G},$ let the set of datasets $x^\prime  \in \dat^{\boldsymbol{d}} \cap \mathcal{G}$ that differ from $x$ on a single entry be denoted by $\textup{N}(x).$ Also, let $T[l] := \{u \in \nn \mid \geopropStrat_{l,u} > 0\}.$ The output of all the DP implementing mechanisms is observationally equivalent to the set of finite output elements, or $\{(A[l][u,\cdot] \otimes B[l]) \boldsymbol{x} + \boldsymbol{y}_{l,u}\}_{l, u\in T[l]},$ because, independently of the data, the remaining DP mechanism answers are infinite with probability one. Thus, without loss of generality, we can restrict our attention to an alternative mechanism that outputs $\{(A[l][u,\cdot] \otimes B[l]) \boldsymbol{x} + \boldsymbol{y}_{l,u}\}_{l, u\in T[l]}.$ Let $\widehat{A}[l] := \stack(\{A[l][u, \cdot]\}_{u\in T[l]}),$ so the strategy matrix for this alternative mechanism can be written as, 

$$ D :=  \stack(\{\widehat{A}[l] \otimes B[l] \}_l). $$

To prove the first result of the theorem, let $\boldsymbol{\hat{\geopropStrat}[l]}:=\stack(\{\geopropStrat_{l,u}\}_{l, u\in T[l]}),$ and $\hat{\boldsymbol{\geopropStrat}}:=\stack(\{\boldsymbol{\hat{\geopropStrat}[l]}\}_l).$ In this case, the noise for this mechanism is distributed as either $\hat{\boldsymbol{y}} \sim \textup{Laplace}(\mathbf{0}, \diag(\boldsymbol{\hat{b}}))$ or $\hat{\boldsymbol{y}} \sim \textup{Laplace}_\zz(\mathbf{0}, \diag(\boldsymbol{\hat{b}})),$ where $\boldsymbol{\hat{b}}:=\stack(\{\boldsymbol{\hat{b}}[l]\}_l)$ and $\boldsymbol{\hat{b}}[l]:=\textbf{2}\oslash(\epsilon(\boldsymbol{\hat{\geopropStrat}[l]}\otimes\boldsymbol{\queryprop}[l])).$ Thus, we will show the following condition, which, by the first result in Lemma 3, will imply the first result of the theorem,

\begin{align} \label{goal_ineq_in_opt_spine_dp_thm}
    \max_{x \in \dat^{\boldsymbol{d}} \cap \mathcal{G}, \; x' \in \textup{N}(x)} \; \sum_i \left| D[i,\cdot] (\boldsymbol{x} - \boldsymbol{x}') /\boldsymbol{\hat{b}}[i] \right| \leq  \epsilon .
\end{align}

\noindent
Starting from the left hand side of this inequality, we have,

\begin{align*}
    \max_{x \in \dat^{\boldsymbol{d}} \cap \mathcal{G}, \; x' \in \textup{N}(x)} \; \sum_i \left| \left(\diag(\mathbf{1} \oslash \boldsymbol{\hat{b}}) D\right)[i,\cdot] (\boldsymbol{x} - \boldsymbol{x}') \right| \leq
    2 \max_j \; \sum_i  \left(\diag(\mathbf{1} \oslash \boldsymbol{\hat{b}}) D\right)[i,j] = &\\
    2 \max_j \; \sum_i  \stack\left(\{ \diag \left(\epsilon(\boldsymbol{\hat{\geopropStrat}[l]}\otimes\boldsymbol{\queryprop}[l])/2\right) (\widehat{A}[l] \otimes B[l]) \}_l \right)[i,j] = & \\
    \epsilon \max_j \; \sum_l \sum_i \left( \diag(\boldsymbol{\hat{\geopropStrat}[l]}) \widehat{A}[l] \otimes \diag(\boldsymbol{\queryprop}[l]) B[l] \right)[i,j]. & \numberthis \label{pt_of_convergence_for_pure_dp}
\end{align*}

\noindent
Note that each of the row sums of $\diag(\boldsymbol{\queryprop}[l]) B[l],$ \textit{i.e.}, the vector $B[l]^\top \diag(\boldsymbol{\queryprop}[l]) \mathbf{1} = B[l]^\top \boldsymbol{\queryprop}[l],$ is equal to $\mathbf{1}$ because $\sum_k Q[i,l][k,\cdot]=\mathbf{1}^\top$ and $\sum_k \queryprop_{k,l}=1.$ Thus, since each column $j$ of $\diag(\boldsymbol{\hat{\geopropStrat}[l]}) \widehat{A}[l]$ contains at most one nonzero value, we have,

\begin{align*}
    \epsilon \max_j \; \sum_l \sum_i \left( (\diag(\boldsymbol{\hat{\geopropStrat}[l]}) \widehat{A}[l]) \otimes (\diag(\boldsymbol{\queryprop}[l]) B[l]) \right)[i,j] = & \\
    \epsilon \max_j \; \sum_l \sum_i \left(  \diag(\boldsymbol{\hat{\geopropStrat}[l]}) \widehat{A}[l]\right)[i,j] =
    \epsilon \max_j \; \sum_i \left( \diag(\hat{\boldsymbol{\geopropStrat}}) \widehat{A}\right)[i,j],&
\end{align*}

\noindent
Note that $\sum_i \left( \diag(\hat{\boldsymbol{\geopropStrat}}) \widehat{A}\right)[i,j]$ is equal to the sum of the geolevel PLB proportions along a path from the US geounit to the block geounit $j.$ Since the bypass operation used within Algorithm \ref{detUpdate} does not change this sum, and its initial value is $\sum_l \geoprop_{l}=1,$ we have

$$    \epsilon \max_j \; \sum_i \left( \diag(\hat{\boldsymbol{\geopropStrat}}) \widehat{A}\right)[i,j] = \epsilon \sum_l \geoprop_{l} =  \epsilon, $$
\noindent
which implies inequality (\ref{goal_ineq_in_opt_spine_dp_thm}).

Similar logic also implies the second result of the theorem. Specifically, in this case we can use the second result of Lemma \ref{composition}, and the fact that for each $x \in \dat^{\boldsymbol{d}} \cap \mathcal{G}$ and $x' \in \textup{N}(x)$ we have $(D[i,\cdot] (\boldsymbol{x} - \boldsymbol{x}') )^2=\left| D[i,\cdot] (\boldsymbol{x} - \boldsymbol{x}') \right|,$ to derive the sufficient condition,

\begin{align} \label{second_goal_ineq_in_opt_spine_dp_thm}
    \max_{x \in \dat^d, \; x' \in \textup{N}(x)} \; \sum_i \left| D[i,\cdot] (\boldsymbol{x} - \boldsymbol{x}') \right| / \boldsymbol{\hat{\sigma}}_i^2 \leq 2 \rho,
\end{align}

\noindent
where $\boldsymbol{\hat{\sigma}^2}:=\stack(\{\boldsymbol{\hat{\sigma}^2[l]}\}_l)$ and $\boldsymbol{\hat{\sigma}^2[l]}:=\mathbf{1} \oslash (\rho (\boldsymbol{\hat{\geopropStrat}[l]} \otimes  \boldsymbol{\queryprop[l]})).$ This inequality follows from,

\begin{align*}
    \max_{x \in \dat^{\boldsymbol{d}} \cap \mathcal{G}, \; x' \in \textup{N}(x)} \; \sum_i \left| D[i,\cdot] (\boldsymbol{x} - \boldsymbol{x}') \right|/\boldsymbol{\hat{\sigma}^2}_i \leq
    2 \max_j \sum_i \left(\diag(\mathbf{1}\oslash \boldsymbol{\hat{\sigma}^2}) D\right)[i,j] = \\
    2 \rho \max_j \sum_i  \left(\diag( \boldsymbol{\hat{\geopropStrat}[l]} \otimes  \boldsymbol{\queryprop[l]}) D \right)[i,j].  \numberthis \label{pt_of_convergence_for_approx_dp}
\end{align*}
\noindent
Since this last value is equal to the value of (\ref{pt_of_convergence_for_pure_dp}) multiplied by $2\rho/\epsilon,$ and the logic above implies the value of (\ref{pt_of_convergence_for_pure_dp}) is equal to $\epsilon,$ we have,

$$ 2 \rho \max_j \sum_i  \left(\diag( \boldsymbol{\hat{\geopropStrat}[l]} \otimes  \boldsymbol{\queryprop[l]}) D \right)[i,j] = \epsilon \cdot (2 \rho/\epsilon) = 2 \rho,$$
\noindent
which implies the sufficient condition (\ref{second_goal_ineq_in_opt_spine_dp_thm}).
\end{proof}

\section{Spine Settings Used for 2020 Census Production Executions} \label{settings}

In this section we will describe the geographic spines used in the 2020 Census data product production executions, which include files for the persons and units universes for both the Redistricting Data (P.L. 94-171) Summary File and the Demographics and Housing Characteristics File (DHC) data products. 

For the redistricting data product persons production execution, the geolevels included in the spine were US, state, county, tract, optimized block group, and block.\footnote{While we focus our attention on the settings for US executions, similar settings were used for the Puerto Rico executions, with the exception that the PR root is at the same level as state in the US hierarchy. The PLB allocated to each geolevel in the Puerto Rico executions was also normalized to ensure the global PLB of each Puerto Rico execution was the same as the global PLB of each corresponding US execution.} The same approach described above was used to include both an AIAN and a non-AIAN branch to this spine at the state geolevel and below. Likewise, optimized block groups were defined using the approach described above using the following four categories of OSEs. First, each AIAN OSE is composed of an individual AIAN area, along with one additional OSE defined by the region outside of all AIAN areas. Second, each GQ OSE is defined as the union of all blocks that contain the same combination of major GQ types. For example, one GQ OSE is the union of all blocks that only contain GQs that are college/university student housing. This OSE category was included to decrease the impact that blocks with GQs have on their neighbors. Third, each minor civil division (MCD) OSE is composed of an MCD in the twelve strong-MCD states, \textit{i.e.}, Connecticut, Maine, Massachusetts, Michigan, Minnesota, New Hampshire, New Jersey, New York, Pennsylvania, Rhode Island, Vermont, and Wisconsin, along with one additional OSE defined by the region outside of all of these MCDs. Fourth, each place OSE is defined by an incorporated or census-designated place that are outside of the twelve strong MCD states, along with one additional OSE that is defined by the region outside of these places. As described above in more detail, optimized block groups were defined by grouping together blocks that are within the same AIAN, GQ, MCD, and place OSEs as well as the same census tract. After all geounits on the spine are defined in this way, the bypassing approach described in Algorithm \ref{detUpdate} was used to define the final PLB values for each geounit.

The spine used in the redistricting data product housing unit production execution was nearly identical to the spine used in the redistricting data product persons production execution with the exception that the GQ OSE category was not considered when defining optimized block groups because occupied GQs are excluded from the universe of the housing units.

The 2020 production DHCP and DHCH executions both use identical spines. The geolevels included on this spine are US, state, county, prim, tract subset group, tract subset, optimized block group, and block. The US, state, county, and block geolevels are defined in the same way as the redistricting data product executions, so next we will define the remaining geolevels. First, prim geounits are defined as the most granular geography unit used internally by the Population Estimates and Projections (PEP) Area of the Census Bureau. All tabulations published by PEP can be derived by aggregation using prims as the aggregation atom.\footnote{These geounits are most commonly called \textit{primitives} in PEP, rather than \textit{prims}. Since the spine used within the DAS also includes more granular(/primitive) geounits than prims, we use prims instead in this paper to avoid confusion.} For example, PEP produces estimates for four places in Autauga County, AL, \textit{i.e.}, Autaugaville town, Billingsley town, Millbrook city, and Prattville city, so the intersection of each of these places and Autauga County defines one prim geounit that is included in this geolevel. In addition, since estimates for each county in the US are also published by PEP, the final prim within Autauga County is defined as the area within Autauga County that is outside of these first four prims. Second, each tract subset geounit was defined as the intersection of a census tract, a prim geounit, and an AIAN OSE, as defined above. Third, tract subset groups were defined by grouping together tract subsets using the same approach described above for grouping blocks to define optimized block groups.

Fourth, optimized block groups were defined in a similar way as in the redistricting spine, but with a different choice of categories of OSEs, which we will define next. Specifically, each school district (SD) OSE is composed of an individual SD, along with one additional OSE defined by the region outside of all SDs. Also, each conventional block group (CBG) OSE is composed of a block group on the standard census geographic spine, along with one additional OSE defined by the region outside of all CBGs. Afterward, optimized block groups were defined by grouping together blocks within the same GQ OSE, tract subset, SD OSE, and CBG OSE. Like the case of the redistricting spines, after all geounits were defined, the bypassing approach described in Algorithm \ref{detUpdate} was used to define the final PLB values for each geounit.

\section{Impact of Spine Optimization: Summary Metrics} \label{sec:metrics}

To demonstrate the impact of spine optimization, this section provides summary metrics of the total population query for three redistricting person-level DAS executions. The 2010 CEF is used as the input data file for all three DAS executions. The execution labeled ``Optimized'' uses settings that are identical to those of the 2020 production redistricting person-level DAS execution. The execution labeled ``Conventional'' uses the production settings with the exception that the internal spine is defined by the official 2010 tabulation geography. The execution labeled ``AIAN'' uses the production settings with the exception that the internal spine is defined by the version used to produce the official 2020 redistricting data.

The tables report the mean absolute error (MAE) of the total population query for various geographic levels and for each of these three DAS executions. Each such MAE value is defined by first finding the total population error of each geounit in the geolevel; afterward, the MAE is defined as the arithmetic mean of the absolute value of these errors. While we focus on the accuracy of the total population query in this section, we have found that the relative comparisons of MAEs between spine settings typically do not strongly depend on the query being considered, since the PLB allocated to the total population query in a given geounit is proportional to the PLB allocated to any alternative query group. 

Table \ref{table:onSpineMAEs} provides total population MAE values for geolevels below the county in the conventional spine. Since we include a state geolevel total population invariant, the total population MAE for the state and US geolevel would be zero, so we only include the MAE for geolevels below the state geolevel. Regarding the column containing the count of the geounits in each geolevel, note that the universe of geounits that we consider only includes the geounits from the conventional spine containing at least one housing unit and/or at least one occupied GQ, which lowers the values of these counts in some cases. Table \ref{table:oseMAEs} provides total population MAE values for OSEs, including the OSEs that were targeted in the spine optimization routines, as described in Section \ref{settings} in more detail. The row ''Aggregated AIAN Areas in States" in this table provides the total population MAE values, averaged over the geographic entities defined as the union of all AIAN blocks within each state. Note that the number of geographic entities is given by 36 because there are 36 states that contain AIAN blocks. The next two rows, ''Aggregated AIAN Areas in Counties" and ''Aggregated AIAN Areas in Tracts" are defined similarly.

\begin{table}[h]
\caption{Total Population query MAE values of US redistricting persons DAS execution for geographic units on the standard census spine.  } \label{table:onSpineMAEs}
\begin{tabular}{lllll}
   &     & \multicolumn{3}{c}{Total Population MAE} \\ \cline{3-5} 
   &     & \multicolumn{3}{c}{Spine Type}  \\ \cline{3-5} 
Geolevel & Geounit Count & Conventional & AIAN    & Optimized    \\
\hline
County & 3,143   & 1.829       & 1.943     &   1.867   \\
Tract & 72,544   & 1.986       & 1.989     &     1.949  \\
Block Group & 216,886   & 1.407       & 1.409     &     15.99  \\
Block & 6,398,202   & 5.006       & 5.007     &     4.849  \\
\hline
\end{tabular}
\end{table}

\begin{table}[h]
\caption{Total Population query MAE values of a US redistricting persons DAS execution for off-spine entities.} \label{table:oseMAEs}
\begin{tabular}{lllll}
   &     & \multicolumn{3}{c}{Total Population MAE} \\ \cline{3-5} 
   &     & \multicolumn{3}{c}{Spine Type}  \\ \cline{3-5} 
Geolevel & Geounit Count & Conventional & AIAN    & Optimized    \\
\hline
AIAN Areas & 621   & 29.19       & 9.741    &   1.876 \\
All Aggregated AIAN Areas & 1 & 6428       & 15.00    &     3.000  \\
Aggregated AIAN Areas in States & 36  & 259.4       & 0.806     &     0.639  \\
Aggregated AIAN Areas in Counties & 412   & 29.09     & 1.629    &     1.391 \\
Aggregated AIAN Areas in Tracts & 1,532   & 11.08      & 1.793     &     1.656 \\
AIAN Block & 165,647   & 4.491      & 4.465     &     4.256 \\
Place\upstairs{\affilone} & 29,250  & 35.36     & 35.26    &     2.679 \\
MCD\upstairs{\affiltwo} & 11,914  & 14.74     & 14.78    &     8.254 \\
\hline
\multicolumn{5}{l}{\footnotesize{The universe of AIAN considered includes all AIAN areas other than state and tribal designated}} \\
\multicolumn{5}{l}{\footnotesize{statistical areas.}} \\
\multicolumn{5}{l}{\footnotesize{\upstairs{\affilone}Place refers to both census designated and incorporated places outside of the 12 strong-MCD states.}} \\
\multicolumn{5}{l}{\footnotesize{\upstairs{\affiltwo}MCD refers to all MCDs in each of the 12 strong-MCD states.}} \\
\end{tabular}
\end{table}

The accuracy metrics in these tables demonstrate several impacts of the DAS internal spine on accuracy that are worth highlighting. For example, Table \ref{table:onSpineMAEs} provides evidence that, relative to the execution that used the conventional spine for the DAS internal spine, the DAS execution that used the AIAN spine exhibited very slightly lower accuracy for the geounits on the conventional spine. In addition to random variability between DAS executions, this change is likely caused by the OSEDs of the geounits in the conventional spine increasing when moving from a DAS internal spine defined as the conventional spine to the AIAN spine. However, this impact is not large because most of the OSEDs of the geounits in the conventional spine are not impacted, and, since each geounit on the conventional spine can be defined by combining at most two geounits on the AIAN spine, the OSEDs that are impacted only increase from one to two. In contrast to these very small changes, Table \ref{table:oseMAEs} shows that moving from the conventional spine to the AIAN spine resulted in a large improvement in the accuracy within AIAN areas, since these areas are much closer to the AIAN spine than the conventional spine. 

Several MAE values in Table \ref{table:onSpineMAEs} for the optimized spine are also noteworthy. First, since using the optimized spine moves conventional (tabulation) block groups further from the spine, the block group geounits in the conventional spine become much less accurate. Second, in contrast, the remaining geolevels improve in accuracy relative to the AIAN spine DAS execution, and, in most cases, also relative to the conventional spine DAS execution. These improvements are a result of using the bypassing step of the spine optimization routines described in Section \ref{bypassing} and also of reducing the largest fanout values, which, as described above, is defined as the number of child geounits of a parent geounit, at the tract geolevel and below. Likewise, all of the MAEs for the optimized spine DAS execution in Table \ref{table:oseMAEs} are lower than the corresponding values of the AIAN spine and the conventional spine DAS executions, which is also a result of using the bypassing step, of decreasing the largest fanout values at the tract geolevel and below, and of bringing these target OSEs closer to the spine. Since block group geounits in the conventional spine are not legally or politically defined, do not correspond to functioning governmental units, and are typically used in cases in which a non-specific geolevel is required with a granularity between that of census tracts and census blocks, the Census Bureau's redistricting tuning experiments pointed to the accuracy improvements provided by spine optimization outweighing the cost of decreased accuracy in conventional (tabulation) block groups.



\bibliographystyle{plain}
\bibliography{references}

\end{document}